\pdfoutput=1
\documentclass[11pt,a4paper]{amsart}
\usepackage[a4paper,inner=2.5cm,outer=2.5cm,top=2.5cm,bottom=2.5cm]{geometry}
\usepackage{amsmath,amssymb,amsthm,enumerate,mathtools,stmaryrd}

\usepackage{xypic}

\usepackage{hyperref}
\hypersetup{colorlinks=true,linkcolor=blue,citecolor=teal,filecolor=magenta,urlcolor=cyan}

\usepackage[textsize=footnotesize]{todonotes}
\presetkeys%
{todonotes}%
{}{}%
\usepackage{xcolor}

\newcommand{\C}{\mathbb{C}}

\usepackage{url}

\usepackage[backend=bibtex,style=alphabetic,sorting=nyt,isbn=false,url=false,doi=true,maxalphanames=10,minalphanames=4,mincitenames=4,maxcitenames=10,minnames=4,maxnames=10,giveninits=true,maxbibnames=99]{biblatex}
\AtEveryBibitem{\clearlist{language}}
\DefineBibliographyStrings{english}{page = {},pages = {}}
\renewbibmacro{in:}{\ifentrytype{article}{}{\printtext{\bibstring{in}\intitlepunct}}}
\DeclareFieldFormat[article,incollection,misc]{title}{\rmfamily{#1}}
\DeclareFieldFormat{addendum}{Preprint version: \urlstyle{rm} \url{#1}}
\theoremstyle{plain}
\newtheorem{theorem}{Theorem}[section]
\newtheorem{proposition}[theorem]{Proposition}
\newtheorem{corollary}[theorem]{Corollary}
\newtheorem{lemma}[theorem]{Lemma}

\theoremstyle{definition}
\newtheorem{definition}[theorem]{Definition}

\makeatletter
\@addtoreset{proofpart}{theorem}
\makeatother
\theoremstyle{remark}
\newtheorem{remark}[theorem]{Remark}

\newcommand{\Z}{\mathbb{Z}}

\newcommand{\bK}{\mathbb{K}}

\newcommand{\calD}{\mathcal{D}}
\newcommand{\cS}{\mathcal{S}}

\newcommand{\cW}{\mathcal{W}}
\newcommand{\bW}{\mathbb{W}}

\newcommand{\res}{\mathop{\rm res}}

\newcommand{\ii}{\mathrm{i}}

\DeclareMathOperator{\Tr}{\mathrm{Tr}}

\newcommand{\coeff}[2]{\mathop{[{#1}^{#2}]}}
\newcommand{\restr}[2]{\mathop{\big\lfloor_{{#1}\to {#2}}}}

\numberwithin{equation}{section}

\title{KP integrability through the $x-y$ swap relation}

\author[A.~Alexandrov]{A.~Alexandrov}
\address{A.~A.: Center for Geometry and Physics, Institute for Basic Science (IBS), Pohang 37673, Korea
}
\email{alex@ibs.re.kr}

\author[B.~Bychkov]{B.~Bychkov}
\address{B.~B.: Department of Mathematics, University of Haifa, Mount Carmel, 3498838, Haifa, Israel}
\email{bbychkov@hse.ru}

\author[P.~Dunin-Barkowski]{P.~Dunin-Barkowski}
\address{P.~D.-B.: Faculty of Mathematics, National Research University Higher School of Economics, Usacheva 6, 119048 Moscow, Russia; HSE--Skoltech International Laboratory of Representation Theory and Mathematical Physics, Skoltech, Bolshoy Boulevard 30 bld. 1, 121205 Moscow, Russia; and NRC “Kurchatov Institute” -- ITEP, 117218 Moscow, Russia}
\email{ptdunin@hse.ru}

\author[M.~Kazarian]{M.~Kazarian}
\address{M.~K.: Faculty of Mathematics, National Research University Higher School of Economics, Usacheva 6, 119048 Moscow, Russia; and Igor Krichever Center for Advanced Studies, Skoltech, Bolshoy Boulevard 30 bld. 1, 121205 Moscow, Russia}
\email{kazarian@mccme.ru}

\author[S.~Shadrin]{S.~Shadrin}
\address{S.~S.: Korteweg-de Vries Institute for Mathematics, University of Amsterdam, Postbus 94248, 1090GE Amsterdam, The Netherlands}
\email{S.Shadrin@uva.nl}	

\begin{document}
	
\begin{abstract} We discuss a universal relation that we call the $x-y$ swap relation, which plays a prominent role in the theory of topological recursion, Hurwitz theory, and free probability theory. We describe in a very precise and detailed way the interaction of the $x-y$ swap relation and KP integrability. As an application, we prove a recent conjecture that relates some particular instances of topological recursion to the Mironov--Morozov--Semenoff matrix integrals.
\end{abstract}
	
\maketitle

\setcounter{tocdepth}{2}
\tableofcontents


\part{Formulations and results}

\section{Pre-introduction}

This paper deals with the so-called $x-y$ swap relation as the main object of study. This relation, in its several different reincarnations, appears to be ubiquitous in various areas of mathematics~\cite{alexandrov2022universal,borot2021analytic,BGF-conjecture,borot-comb,FullySimpleProof,borot2021topological,hock2022simple,hock2022xy,hock2023laplace}. We discuss in details the structure of this relation, representing it in several different ways, and prove that this relation preserves KP integrability as a general property of systems of differentials.

In order to present the results of the paper, we have decided to split them in three independent layers. Firstly, we just introduce in Part 1 of the Introduction (Section~\ref{sec:intro1}) in a precise but very rough way the main results that can be understood without going deeply into the structure of the $x-y$ swap relation and/or KP integrability. Then, in Part 2 of the Introduction (Section~\ref{sec:intro2KP}) we go deeper inside the structures related to KP integrability and formulate the results on their transformation rules under the $x-y$ swap relation. Finally, in Part 3 of the Introduction (Section~\ref{sec:intro3FullPicture}) we present a system of internal constructions and equivalent reformulations of $x-y$ symmetry that all together fully revisit the $x-y$ swap relation and connect it to KP integrability. Importantly, these reformulations avoid the complicated combinatorics of sums over graphs which appear in the original formulation of the $x-y$ swap relation.

One of the basic principles that we formulate in this paper is the concept of KP integrability as a property of a system of symmetric meromorphic differentials $\{\omega^{(g)}_n\}_{g\geq 0, n\geq 1}$ that are all regular along the diagonal, except for the $(g,n)=(0,2)$ case. This principle is stated in Theorem~\ref{thm:KPglobal} as an assertion that the KP integrability property does not depend on the point of expansion.

The main results of the paper are Theorem~\ref{thm:KP-duality} that basically states that the $x-y$ swap relation preserves KP integrability; Theorem~\ref{thm:transformation-KKernel} that gives a Gaussian integral transformation for the corresponding Baker--Akhiezer kernels; and Theorem~\ref{thm:transformation-Omega} that represents a deep structure behind the $x-y$ swap relation that is responsible for the transfer of integrability properties.

The main applications concern the questions of KP integrability in the theory of spectral curve topological recursion~\cite{EO,B-E}: Theorem~\ref{thm:r-BGW} resolves a conjecture in~\cite{alexandrov2023higher,chidambaram2023relations} that a particular instance of topological recursion reproduces the Mironov--Morozov--Semenoff matrix integral, and Theorem~\ref{thm:TR-main} describes a quite general situation when our theory provides an explicit description of the KP tau function associated to the correlation differentials of topological recursion. We also prove that in the context of topological recursion the KP integrability property requires that the spectral curve is rational, see Theorem~\ref{thm:TR-rational-curve-refined}.

Beyond applications in topological recursion we have a quite general statement on the particular representation as Gaussian integrals for the basis vectors that span the half-infinite planes (or, equivalently, points in the Sato--Wilson Grassmannian) corresponding to tau functions obtained by application of $x-y$ swap relation to the trivial case, see Corollary~\ref{cor:Phi-TirivalDual}.

\subsection{Notation} $\ii$ denotes $\sqrt{-1}$. For a function $f(x)$ the expression  $\restr{x}{a} f(x)$ stands for $f(a)$. By $\llbracket n \rrbracket$ we mean the set $\{1,\dots,n\}$, and $x_{\llbracket n \rrbracket}$ denotes the set of variables $\{x_1,\dots,x_n\}$.

\subsection{Organization of the paper} Sections~\ref{sec:intro1}--\ref{sec:intro3FullPicture} form together an introduction, where we formulate our results as described above. We supply our statements with immediate proofs if they are sufficiently short and do not use involved computations with the Gaussian integrals (the latter integrals are the main technique we use in this paper). In Section~\ref{sec:quasiclassical} we prove the results that concern the semi-classical limit of the corresponding KP tau functions (which are, in particular, necessary to describe the possible form of the $(0,2)$ differentials in KP integrable case). In Section~\ref{sec:Gaussian} we develop the technique of computations with the Gaussian integrals and their connection to the ingredients of the $x-y$ swap relations. This section is the technical core of this paper. Finally, in Section~\ref{sec:ReformulationXY-duality} we present necessary details to complete the proof of identities of Section~\ref{sec:intro3FullPicture}.

It is important to mention that we assume that the reader is a bit familiar with Kadomtsev--Petviashvili hierarchy, which we call elsewhere KP hierarchy for brevity. Otherwise we refer to~\cite{MJD} as the main source on this topic. Also we do not recall at all the concept of spectral curve topological recursion. As we explain below, in many cases it is sufficient to know that it can be replaced by the $x-y$ swap formula, and otherwise we refer the reader to the foundational survey~\cite{EO} and more recent sources (e.g.~\cite{alexandrov2022universal}) for precise statements that we use.

\subsection{Acknowledgments} A.A.  was supported by the Institute for Basic Science (IBS-R003-D1). Research of B.B. was supported by the ISF Grant 876/20. P.D.-B. and M.K. were supported by the International Laboratory of Cluster Geometry NRU HSE, RF Government grant, ag. № 075-15-2021-608 dated 08.06.2021. S.S. was supported by the Netherlands Organization for Scientific Research. A.A., P.D.-B., M.K., and S.S. are grateful for the hospitality to the University of Haifa, where part of this research was carried out. A.A. and S.S. thank IHES for hospitality as well. A.A. thanks the mathematical research institute MATRIX in Australia where part of this research was performed. B.B and S.S. thank the Erwin Schr\"odinger  International Institute for Mathematics and Physics in Austria where part of this research was performed. We
are also grateful to the anonymous referee for the helpful remarks and suggestions.

\section{Introduction, part 1: first layer of concepts and results} \label{sec:intro1}

\subsection{\texorpdfstring{$n$}{n}-point differentials} \label{sec:n-pt-diff} Let $\Sigma$ be a smooth complex curve that we refer below as the \emph{spectral curve}. The main object of this paper is a collection of $n$-differentials $\omega^{(g)}_n$ on $\Sigma^n$ defined for all $g\ge0$, $n\ge1$.

We assume that all $\omega^{(g)}_n$'s are symmetric and meromorphic with no poles on the diagonals for $(g,n)\ne(0,2)$, and $\omega^{(0)}_2$ is also symmetric and meromorphic but it has a second order pole on the diagonal with biresidue~$1$.
It will be convenient to arrange the $n$-differentials $\omega^{(g)}_n$, $g\geq 0$, for each fixed $n$ into generating series:
\begin{equation}\label{eq:omega-allgenus}
	\omega_n\coloneqq \sum_{g=0}^\infty\hbar^{2g-2+n}\omega^{(g)}_n.
\end{equation}

Systems of $n$-point differentials provide a natural and convenient way to assemble the answers to various problems in many interrelated areas of mathematics and theoretical physics, for instance, in matrix models, enumerative algebraic geometry, and combinatorics.

\subsection{KP integrability} KP integrability plays a prominent role in mathematics and mathematical physics, the basic reference is~\cite{MJD}.

The KP hierarchy can be seen as an integrable system of PDEs for a so-called tau function that can be written in several different ways. For this paper it is important that for formal solutions it has an interpretation as Pl\"ucker relations for the open part of semi-infinite Grassmannian.  These relations ensure that a particular vector $v\in \bigoplus_{\lambda} v_\lambda$ (here $v_\lambda = z^{\lambda_1-1}\wedge z^{\lambda_2-2}\wedge z^{\lambda_3-3}\wedge\dots$  and the sum  is taken over all partitions $\lambda\vdash d$, $d \in\Z_{\geq 0}$ arranged as $\lambda_1\geq \lambda_2\geq\lambda_3\geq ...$) can be represented as
\begin{align} \label{eq:decomposable-vector}
v= \Phi_1 \wedge \Phi_2 \wedge \Phi_3 \wedge \dots,	
\end{align}
where $\Phi_i = z^{-i}(1+O(z))$. For a vector $v$, the corresponding tau function $\tau_v$ can be written in formal variables $p_1,p_2,\dots$ as
\begin{align}
	\tau_v\coloneqq \langle 0| \exp\left({\sum_{i=1}^\infty} \frac{p_i }{i} J_i \right) v.	
\end{align}
Here $J_a = \sum_{i\in\Z}:\psi_{i-a}\psi^*_i:$, where $\psi_i= z^i\wedge$, $\psi^*_i = \partial_{z^i}$, the normal ordering is defined as $:\psi_{i}\psi^*_j: = \psi_{i}\psi^*_j$ for $j\geq 0$ and $-\psi^*_j\psi_{i}$ for $j< 0$, and  $\langle 0|$ is the operator which gives a constant term of the formal power series in variables $p_1,p_2,\ldots.$ The condition given by Equation~\eqref{eq:decomposable-vector} is then equivalent to the Hirota bilinear identity
\begin{align}
	\oint_\infty e^{\sum_{i=1}^\infty \frac{p_i-q_i}{i} z^i} e^{\sum_{i=1}^\infty z^{-i} (\partial_{q_i}-\partial_{p_i})} \tau_v(p_1,p_2,\dots) \tau_v(q_1,q_2,\dots) dz =0.
\end{align}

There is a version of it that takes into account topological expansions of the corresponding $\tau$ functions, which is sometimes called the $\hbar$-KP integrability, see~\cite{TakasakiTakebe95}.

\subsection{KP integrability as a property of a system of differentials} To pass from a system of differentials to KP tau functions, we need the following definition:

\begin{definition}
A point $o\in\Sigma$ is called \emph{regular} for the system of differentials $\{\omega^{(g)}_n\}$ if $\omega^{(g)}_n-\delta_{(g,n),(0,2)}\frac{dx_1dx_2}{(x_1-x_2)^2}$ is regular at $(o,\dots,o)\in\Sigma^n$ for all $(g,n)\in\Z_{\geq 0}\times\Z_{>0}$, $(g,n)\not=(0,1)$, where~$x$ is a local coordinate on $\Sigma$ at~$o$.
\end{definition}

\begin{remark}
	Note that the regularity condition for $\{\omega^{(g)}_n\}$ is independent of a choice of local coordinate. 	Note also that we have no condition for $\omega^{(0)}_1$.
\end{remark}

For a regular point~$o$, and an arbitrary local coordinate~$x$, the forms $\omega^{(g)}_n$ can be expanded into a series
\begin{equation}\label{eq:omega-expansion}
	\omega_n=\sum_{g=0}^\infty\hbar^{2g-2+n}\omega^{(g)}_n=\delta_{n,2}\tfrac{dx_1dx_2}{(x_1-x_2)^2}+
	\sum_{k_1,\dots,k_n=1}^{\infty}f_{k_1,\dots,k_n}\prod_{i=1}^n k_i x_i^{k_i-1} dx_i,
\end{equation}
where the coefficients $f_{k_1,\dots,k_n}$ expand as $\sum_{g=0}^\infty \hbar^{2g-2+n} f^{(g)}_{k_1,\dots,k_n}$.
Introduce the associated potential $F=F_{o,x}$ as
\begin{equation}\label{eq:F-expansion}
	F(p_1,p_2,\dots)\coloneqq\sum_{g,n}\hbar^{2g-2+n}F^{(g)}_{n}=\sum_{n=1}^\infty\frac{1}{n!}
	\sum_{k_1,\dots,k_n=1}^{\infty}f_{k_1,\dots,k_n}p_{k_1}\dots p_{k_n},
\end{equation}
where $F^{(g)}_n$ are homogeneous formal series of degree~$n$.

This way we associate to a given collection of differentials $\omega^{(g)}_n$ a large family of potentials $F_{o,x}$: the freedom in its definition consists of the choice of a regular point $o$ and the choice of a local coordinate $x$ at $o$.
We would like to study conditions assuring that
\begin{equation}\label{eq:ZFrel}
	Z_{o,x}\coloneqq\exp(F_{o,x})
\end{equation} is a tau function of the KP hierarchy.


\begin{theorem}\label{thm:KPglobal} 
	\emph{The KP integrability is an internal property of the collection of differentials:} if $Z_{o,x}$ is a tau function of KP hierarchy for some choice of a regular point $o$ and a coordinate $x$ at this point, then it is a tau function of KP hierarchy for any other choice of a regular point and a local coordinate at that point.
\end{theorem}

We prove this theorem in Section~\ref{sec:DetFormProofThmKPGlobal}.

\begin{remark} \label{rem:01irrelevant}
	The choice of $\omega^{(0)}_1$ does not affect the KP integrability property, it can be chosen in an arbitrary way. For definiteness, in \eqref{eq:omega-expansion} and in all expressions below which contain a summation over genus we do not include the terms with $(g,n)=(0,1)$.
\end{remark}
\begin{remark}
A change of the local spectral parameter (local coordinate on the spectral curve) is known to provide a symmetry of the KP hierarchy \cite{Shiota}, given by an element of the Virasoro subgroup of $\widehat{GL}_\infty$.
\end{remark}
\subsection{\texorpdfstring{$x-y$}{x--y} swap relation} The $x-y$ swap relation is given by several equivalent formulas that act on a system of differentials $\omega^{(g)}_n$ and produce a new system of differentials denoted by $\omega^{\vee,(g)}_n$. In order to apply it we need to assume that $\Sigma$ is endowed a pair of meromorphic functions on it, denoted by $x$ and $y$.  We present the $x-y$ swap relation here in the form given in~\cite{hock2022simple,alexandrov2022universal}.


As we have mentioned in Remark~\ref{rem:01irrelevant}, the choice of $\omega^{(0)}_1$ is irrelevant for the KP integrability property, so for instance we can fix it to be $\omega^{(0)}_1=-ydx$ following the tradition coming from the theory of spectral curve topological recursion. However, it will be more convenient to fix it as $\omega^{(0)}_1=0$ and to single out the contribution of $-ydx$ in the formulas explicitly.

Then we define a new system of differentials $\omega^{\vee,(g)}_n$ given by
\begin{multline} \label{eq:MainFormulaSimple}
	\frac{\omega_{n}^{\vee,(g)} (z_{\llbracket n\rrbracket})}{\prod_{i=1}^n dy_i} 
	 =	(-1)^n  \coeff \hbar {2g}
    \prod_{i=1}^n  \sum_{r_i=0}^\infty  \partial_{y_i}^{r_i} [u_i^{r_i}]
         \frac{dx_i}{dy_i}\frac{e^{-u_i(\cS(u_i\hbar\partial_{x_i})-1)y_i}}{u_i}
	\\
    \sum_{\Gamma} \frac{\hbar^{2g(\Gamma)}}{|\mathrm{Aut}(\Gamma)|}
	\prod_{e\in E(\Gamma)} \prod_{j=1}^{|e|}\restr{(\tilde u_j, \tilde x_j)}{ (u_{e(j)},x_{e(j)})} \tilde u_j \cS(\hbar \tilde u_j  \partial_{\tilde x_j})
   \sum_{\tilde g=0}^\infty \hbar^{2\tilde g}\frac{\tilde \omega^{(\tilde g)}_{|e|}(\tilde z_{\llbracket |e|\rrbracket})}{  \prod_{j=1}^{|e|}d\tilde x_j}.
\end{multline}
Here $x_i=x(z_i)$, $y_i=y(z_i)$, where $z_i$ is a point on the $i$-th copy of $\Sigma$ in $\Sigma^n$; the same applies to $\tilde x_i = x(\tilde z_i)$. The sum is taken over all connected  bipartite graphs (or hypergraphs or just graphs with multiedges, for brevity)  $\Gamma$ with $n$ labeled vertices and multiedges of arbitrary index $\geq 1$, where the index of a multiedge is the number of its ``legs'' and we denote it by $|e|$. For a multiedge $e$ with index $|e|$  we control its attachment to the vertices by the associated map $e\colon \llbracket |e| \rrbracket \to
\llbracket n \rrbracket
$ that we denote also by $e$, abusing notation (so $e(j)$ is the label of the vertex to which the $j$-th ``leg'' of the multiedge $e$ is attached).
Terms with $(g,n)=(0,1)$ are excluded from the summation, and $\tilde \omega^{(0)}_{2}(\tilde x_1,\tilde x_2) =  \omega^{(0)}_{2}(\tilde x_1,\tilde x_2) - \frac{d\tilde x_1d\tilde x_2}{(\tilde x_1-\tilde x_2)^2}$ if $e(1)=e(2)$, and $\tilde \omega^{(0)}_{2}(\tilde x_1,\tilde x_2) =  \omega^{(0)}_{2}(\tilde x_1,\tilde x_2)$ otherwise. For all $(g,n)\not= (0,2)$ we simply have $\tilde \omega^{(g)}_{n} =  \omega^{(g)}_{n}$.
 By $g(\Gamma)$ we denote the first Betti number of $\Gamma$. By $[\hbar^{2g}]$ (respectively, $[w_i^{k_i}]$) we denote the operator that extracts the corresponding coefficient from the whole expression to the right of it. By $\restr{a}{b}$ we denote the operator of substitution $a\to b$. The function $\cS(z)$ is defined as
\begin{align}
	\cS(z)\coloneqq \frac {e^{z/2}-e^{-z/2}} z.
\end{align}

In the theory of topological recursion, this formula expresses the swap of $x$ and $y$~\cite{alexandrov2022universal}, see also~\cite{hock2022xy} in genus $0$.
It also expresses a higher genera extension of the correspondence between moments and cumulants in the higher order free probability theory~\cite{borot2021analytic}.

\begin{remark}
	Note that it is absolutely non-trivial that the right hand side of \eqref{eq:MainFormulaSimple} does not have poles on the diagonals. It follows from \cite[Proposition 4.6]{alexandrov2022universal} where we proved it using formal expansion in the Fock space.
\end{remark}

\begin{remark} The formulation of the $x-y$ swap relation does not explicitly require that $x$ and $y$ are meromorphic; it is sufficient to assume that $dx$ and $dy$ are meromorphic (so these $1$-forms might have non-trivial residues). It is possible to extend the results of this paper that concern KP integrability to this more general setup, see \cite{ABDKS4}.
\end{remark}

\subsection{\texorpdfstring{$x-y$}{x--y} swap relation and KP integrability}
One of the main results of this paper can be formulated as follows:

\begin{theorem}\label{thm:KP-duality} 
	\emph{The $x-y$ swap relation preserves KP integrability}: the original system of differentials $\omega^{(g)}_n$ is KP integrable (in the sense of Theorem~\ref{thm:KPglobal}) if and only if the dual system of differentials $\omega^{\vee,(g)}_n$ is KP integrable.
\end{theorem}

This theorem is proved in Section~\ref{Sec:proofof2.5}.

On the one hand, this theorem is quite natural since the $x-y$ swap relation is itself coming from an element of $\widehat{GL}_\infty$ that preserves KP integrability of tau functions. Namely, if there exists a point $o\in \Sigma$ regular for the differentials $\{\omega_n^{(g)}\}$ such that functions $x$ and $y$ have at this point a simple zero and a simple pole, respectively, and $x(z)y(z)\to 1$ as $z\to o$,  then (see \cite{alexandrov2022universal}) we have
\begin{align}
	Z^\vee_{o,y^{-1}} = \calD_{1+\theta} Z_{o,x},
\end{align}
where the operator $\calD_{1+\theta}$ defined as a diagonal operator in the basis of Schur polynomials $s_\lambda\coloneqq \tau_{v_\lambda}$ acting as
\begin{align}
	\calD_{1+\theta} \colon s_\lambda \mapsto \Big(\prod_{(i,j)\in\lambda} (1+\hbar(i-j)) \Big) s_\lambda
\end{align}
is known to preserve the KP integrability. On the other hand, in the absence of a regular point with these properties, Theorem~\ref{thm:KP-duality} becomes a very deep statement on the structure of $x-y$ swap relation that requires a full revision of our understanding of the meaning and interpretations of the latter formula. We prove it at the end of Section \ref{Sec:proofof2.5}.

\subsection{Higher BGW tau functions} Topological recursion associates to a particular input that consists of a spectral curve $\Sigma$, functions $x$ and $y$ on $\Sigma$, and a choice of the Bergman kernel, the output that consists of a system of symmetric $n$-differentials $\omega^{(g)}_n$ on $\Sigma^n$. We refer to~\cite{EO} and~\cite{B-E} for the necessary definitions but omit them in this paper for the following reason. In the context of this paper in the special case when $\Sigma=\C P^1$ and $y$ is an affine coordinate, the resulting differentials $\omega^{(g)}_n$ are obtained by the $y-x$ swap relation (that is, the inverse of $x-y$ swap relation, which is given by~\eqref{eq:MainFormulaSimple} with the interchanged $x$ and $y$ and $\omega^{(g)}_n$'s and $\omega^{\vee,(g)}_n$'s, respectively) applied to $\{\omega^{(g),\vee}_n\}_{g\geq 0, n\geq 1}$ defined as $\omega^{(0),\vee}_1= -xdy$, $\omega^{(0),\vee}_2 = dy_1dy_2/(y_1-y_2))^2$, and $\omega^{(g),\vee}_n = 0$ for $2g-2+n>0$, see~\cite[Section 7]{alexandrov2022universal}. This description is sufficient to formulate the following result.

Consider a system of differentials constructed with the input given by $\Sigma=\C P^1$ with a global rational coordinate $z$, $x=z^{r}/r$, and $y=z^{-1}$, and $B =\frac{dz_1dz_2}{(z_1-z_2)^2}$. One can show that $z=\infty$ is a regular point for this system of differentials, so we consider the partition function $Z_{\infty,z^{-1}}$.

\begin{theorem}\label{thm:r-BGW} 
	The partition function $Z_{\infty,z^{-1}}$ is given by the following matrix integral:
	\begin{align}\label{eq:MMBGW}
		Z_{\infty,z^{-1}} = \frac{  \int_{\mathcal{H}_N} d M \exp\bigg(\frac{1}{{\hbar}} \Tr \Big( \frac{\Lambda^r M}{r} +\frac{M^{1-r}}{r(r-1)}- {\hbar}N \ln M\Big)\bigg) }
		{  \int_{\mathcal{H}_N} d M \exp\bigg(\frac{1}{{\hbar}} \Tr \Big( \frac{M^{1-r}}{r(r-1)}- {\hbar}N \ln M\Big)\bigg) },
	\end{align}
	where $\mathcal{H}_N$ is the space $N\times N$ Hermitian matrices and $dM$ is the standard measure on it, and the variables $p_i$ are the Miwa variables, $p_i = \Tr \Lambda^{-i}$, $i=1,2,\dots$.
\end{theorem}

This theorem is proven in Section~\ref{sec:toporec-nonramified-y}.

\begin{remark}
	The latter matrix integral is considered as formal matrix integral, asymptotically expanded in the vicinity of Gaussian point, c.f. \cite{MiMoSeff96, alexandrov2023higher} and Section \ref{sec:FormalGaussianIntegral}.
\end{remark}

\begin{remark}
	The latter matrix integral is considered in~\cite{MiMoSeff96, alexandrov2023higher} and is known as a higher Br\'ezin--Gross--Witten tau function, or simply $r$-BGW tau function.
	The statement of Theorem~\ref{thm:r-BGW} was conjectured in~\cite[Conjecture 4.1]{alexandrov2023higher} and~\cite[Conjecture F]{chidambaram2023relations}.
\end{remark}

\begin{remark}
	Partition function $Z_{\infty,z^{-1}}$ has a geometric meaning. The formal power series expansion of its logarithm gives the intersection numbers of the tautological $\psi$-classes with the so-called $r$-theta classes $\Theta^{r}_{g,n}$ on the moduli spaces of curves $\overline{\mathcal{M}}_{g,n}$~\cite{chidambaram2023relations}. In particular, for $r=2$ it is reduced to a generating series of $\psi$- and a combination of $\kappa$-classes with some distinguished polynomial behavior~\cite{kazarian2021polynomial}.
\end{remark}

\begin{remark}
From the identification with the matrix integral it immediately follows \cite{MiMoSeff96} that $Z_{\infty,z^{-1}}$ is a tau function of the $r$-reduction of the KP hierarchy (or Gelfand--Dickey hierarchy).
\end{remark}

\section{Introduction, part 2: a deeper look into KP integrability and \texorpdfstring{$x-y$}{x-y} swap relation} \label{sec:intro2KP}

\subsection{Baker--Akhiezer kernel and determinantal formulas} For a given formal power series $\tau=\tau(p_1,p_2,\dots)=1+O(p_1,p_2,\dots)$ the kernel $K(x_1,x_2)$ is defined as
\begin{align}\label{eq:K-def}
	K(x_1,x_2) \coloneqq \frac{1}{x_1-x_2} \Big(\prod_{i=1}^\infty \restr{p_i}{x_1^i-x_2^i} \Big) \tau.
\end{align}
If $\tau$ is a tau function of the KP hierarchy, then $K$ is called the Baker--Akhiezer kernel.
It is a sum of one specific term $(x_1-x_2)^{-1}$ and some formal power series in $x_1,x_2$.

Under the assumption of KP integrability it characterizes the solution uniquely, using formuals \eqref{eq:W-K-relation} and \eqref{eq:K-def}, and can be used to express the formal $n$-point functions $W_n$ associated to $\tau$ by the following formulas. Let
\begin{align} \label{eq:W-definition}
	W_n(x_1,\dots,x_n)\coloneqq \restr{p_1,p_2,\dots}{0} \Big(\prod_{i=1}^n \sum_{k=1}^\infty k x_i^{k-1} \partial_{p_k} \Big)\log \tau,
\end{align}
then we have:
\begin{align} \label{eq:W-K-relation}
	W_1 (x_1) & = \lim_{x_1'\to x_1} \Big(K(x_1,x_1')- \frac{1}{x_1-x_1'}\Big); \\ \notag
	W_2(x_1,x_2) & =-K(x_1,x_2)K(x_2,x_1) -\frac{1}{(x_1-x_2)^2}; \\ \notag
	W_n(x_{\llbracket n\rrbracket})  & = (-1)^{n-1} \sum_{\sigma\in C_n} \prod_{i=1}^n K(x_{i},x_{\sigma(i)}), & n\geq 3,
\end{align}
where $C_n\subset S_n$ is the set of cycles of length $n$.
These formulas are called the determinantal formulas and are direct consequences of the Wick formula \cite{Zhou}. More precisely, from~\cite[Section 5]{Zhou} we have:

\begin{lemma}\label{lem:determinantal-formulas} 
	A formal power series $\tau$ is a KP tau function if and only if Equations~\eqref{eq:W-K-relation} hold.
\end{lemma}

\subsection{Determinantal formulas for systems of differentials and a proof of Theorem~\ref{thm:KPglobal}}

\label{sec:DetFormProofThmKPGlobal}

In the context of KP integrability as a property of a system of differentials, we consider a given system of differentials $\{\omega^{(g)}_n\}$, and we assume that the tau function and the corresponding kernel $K$ of the previous section are associated with a particular choice of the regular point~$o$ of the spectral curve and the local coordinate~$x$ at this point. Since linear terms of the potential do not affect KP integrability, let us assume for definiteness that
\begin{equation}\label{eq:assume-omega01-equal0}
\omega^{(0)}_1\coloneqq 0,	
\end{equation}
see Remark~\ref{rem:01irrelevant}.

We introduce an invariant version of the Baker--Akhiezer kernel as the bi-half-differential
\begin{equation}\label{eq:K-bK}
\bK(x_1,x_2)\coloneqq K(x_1,x_2)\sqrt{dx_1 dx_2}.
\end{equation}

\begin{remark} By a half-differential we understand a local section of a line bundle $L$ on $\Sigma$ such that $L^{\otimes 2}\cong K_\Sigma$. A particular choice of $L$ (known as theta characteristic) is fixed and is the same for all variables, and the exposition below doesn't depend on this choice.
\end{remark}

Note that if $\tau=Z_{o,x}$ via~\eqref{eq:omega-expansion}--\eqref{eq:ZFrel} for the system of differentials $\{\omega^{(g)}_n\}$, then the $n$-point functions $W_n$ defined in~\eqref{eq:W-definition} can be expressed in terms of the $\{\omega^{(g)}_n\}$'s themselves as follows:
\begin{align}
	W_n = \frac{\omega_n }{\prod_{i=1}^n dx_i}-\delta_{n,2}\frac{1}{(x_1-x_2)^2},
\end{align}
where $\omega_n= \sum_{g=0}^\infty \hbar^{2g-2+n}\omega^{(g)}_n$.

Then the determinantal formulas~\eqref{eq:W-K-relation} in the KP integrable case can be rewritten as
\begin{align} \label{eq:omega-kappa-relation}
	\omega_1(x_1) & = \lim_{x_1'\to x_1} \Big(\bK(x_1,x_1')-\frac{\sqrt{dx_1dx'_1}}{x_1-x_1'}\Big); \\ \notag
	\omega_n(x_{\llbracket n\rrbracket})  & = (-1)^{n-1} \sum_{\sigma\in C_n} \prod_{i=1}^n
	\bK(x_{i},x_{\sigma(i)}), & n\geq 2.
\end{align}
Note that we consider~\eqref{eq:K-bK}--\eqref{eq:omega-kappa-relation} at the moment as equalities of power expansions at the point~$o$ in the local coordinate~$x$ of the spectral curve.

\begin{remark}
	Equation~\eqref{eq:omega-kappa-relation} is not a completely straightforward rewriting of~\eqref{eq:W-K-relation}; the identification of the formulas for $\omega_1$ and $W_1$ uses a corollary of the KP integrability equations for $\omega^{(0)}_2$ under the assumption~\eqref{eq:assume-omega01-equal0}, cf.~Sections~\ref{sec:KP-quasiclassical} and~\ref{sec:quasiclassical} below.
\end{remark}

\begin{proposition}\label{prop:bK-global}
The kernel~$\bK$ extends as a formal series in~$\hbar$ whose coefficients are global meromorphic differentials on~$\Sigma^2$ defined in a vicinity of the diagonal.
The introduced differential $\bK$ is independent of a choice of the point~$o$ and of the local coordinate~$x$, and is uniquely determined by the system of differentials~$\{\omega^{(g)}_n\}$.
\end{proposition}

\begin{proof}
For the proof we provide an alternative invariant expression for~$\bK$ directly in terms the differentials~$\omega^{(g)}_n$:
\begin{align}\label{eq:bK-def}
 	\bK(z_1,z_2)& \coloneqq \frac{\sqrt{dx_1\;dx_2}}{x_1-x_2}
 \exp \bigg(
 \sum\limits_{2g-2+n>0}\frac{\hbar^{2g-2+n}}{n!}
 		\int\limits_{z_2}^{z_1}\dots\int\limits_{z_2}^{z_1}\omega^{(g)}_n+
 \frac12
 		\int\limits_{z_2}^{z_1}\!\!\int\limits_{z_2}^{z_1}(\omega^{(0)}_2-\omega^{(0),{\rm sing}}_2)\bigg),
 \\\omega^{(0),{\rm sing}}_2&\coloneqq\frac{dx_1dx_2}{(x_1-x_2)^2}.
\end{align}

The integrals $\int_{z_2}^{z_1}$ under the exponent in this expression of $\bK$ are single-valued functions in the vicinity of the diagonal, since we the integral path is chosen in the vicinity of the diagonal. To this end we note that the only possible multi-valuedness can occur in the neighborhood of a simple pole $p$ of $\omega^{(g)}_n$. In this case we choose a branch of the logarithm of $\log{z_1/z_2}$ that is a single valued function near $(p,p)\in\Sigma^2$ once we remove a cut connecting the germs of the divisors $z_1=p$ and $z_2=p$ near the diagonal. We regard $\bK$ as a (bi-half-)differential on~$\Sigma^2$ and use $z_1$ and $z_2$ as a notation for the corresponding points of the spectral curve, and $x_i=x(z_i)$. This formula makes sense if $z_1$ and $z_2$ are close to one another and is actually a reformulation of~\eqref{eq:K-def}.

The definition~\eqref{eq:bK-def} involves a regularizing term $\omega^{(0)\rm, sing}_2$ that does depend on a choice of local coordinate~$x$ --- the equation does not make  sense without this regularizing term. However, it can be easily checked by direct substitution that a change of local coordinate implying the corresponding change of the regularizing term leads to an equivalent equation. Indeed, for any other local coordinate $y$ near $o$, $y=y(x)$, we have:
\begin{equation}\label{eq:omega02eq-solution}
	\frac{dy(x_1)dy(x_2)}{(y(x_1)-y(x_2))^2}=\frac{dx_1dx_2}{(x_1-x_2)^2}\,\exp\left({\int\limits_{x_2}^{x_1}\int\limits_{x_2}^{x_1}
		\left(\frac{dy(x_1)dy(x_2)}{(y(x_1)-y(x_2)^2}-\frac{dx_1dx_2}{(x_1-x_2)^2}\right)}\right).
\end{equation}

If we take for~$x$ any global meromorphic function then the form~$\bK$ defined by~\eqref{eq:bK-def} is also meromorphic in a vicinity of the diagonal in~$\Sigma^2$. This completes the proof.
\end{proof}
For the KP integrable cases below this globally defined differential $\bK$ we also call the Baker--Akhiezer kernel.

\begin{remark}
By construction, the differential~$\bK$ is singular on the diagonal. However, the difference $\bK-\frac{\sqrt{dx_1dx_2}}{x_1-x_2}$ is regular at $x_1=x_2$ for any choice of the local coordinate~$x$.
\end{remark}

\begin{proof}[Proof of Theorem~\ref{thm:KPglobal}]
Theorem~\ref{thm:KPglobal} is a direct corollary of Proposition~\ref{prop:bK-global}. Indeed, given a system of differentials  $\{\omega^{(g)}_n\}$, assume that $Z_{o,x}$ for a particular regular point~$o$ of the spectral curve and a local coordinate~$x$ is a tau function of the KP hierarchy. Then Equations~\eqref{eq:W-K-relation}, and hence,
Equations~\eqref{eq:omega-kappa-relation} hold in the power expansions at the point $(o,\dots,o)\in\Sigma^n$. But both sides of~\eqref{eq:omega-kappa-relation} have global invariant meaning, and hence,~\eqref{eq:omega-kappa-relation} holds true as an equality of global differentials defined in a neighborhood of the diagonal in~$\Sigma^n$. It follows that~\eqref{eq:omega-kappa-relation}, and hence,~\eqref{eq:W-K-relation} hold for the power expansions at any other point regular for all differentials~$\omega^{(g)}_n$ implying KP integrability of the corresponding potential.
\end{proof}

\subsection{KP integrability in the semi-classical limit and application to topological recursion}
\label{sec:KP-quasiclassical}

Consider a system of differentials $\{\omega^{(g)}_n\}$ and assume that they satisfy the KP integrability property in the sense of Theorem~\ref{thm:KPglobal}.
As we pointed out in Remark~\ref{rem:01irrelevant}, the $(0,1)$ term is irrelevant for the KP integrability property. But then KP integrability implies severe restrictions for the $(0,2)$ term (which we call the semi-classical limit) that we formulate now.
\begin{proposition}\label{prop:omega02KP} 
Let $\{\omega^{(g)}_n\}$ be a system of symmetric meromorphic differentials satisfying KP integrability property in the sense of Theorem~\ref{thm:KPglobal}. Let $o$ be a regular point for this system. Then there exists a local coordinate~$z$ at this point such that the $(0,2)$-term of the corresponding potential $F_{o,z}=\log(Z_{o,z})$ vanishes. Respectively, for this choice of the local coordinate we have
	\begin{equation}\label{eq:omega02z}
		\omega^{(0)}_2=\frac{dz_1dz_2}{(z_1-z_2)^2}.
	\end{equation}
Moreover, such local coordinate~$z$ extends as a global meromorphic function on~$\Sigma$ and the above equality holds globally.
	The function~$z$ with this property is defined uniquely up to a linear fractional transformation.
\end{proposition}

For the proof of this Proposition see Section~\ref{sec:quasiclassical}. Its local version is well known and proved, in particular, in~\cite[Proposition B.2]{TakasakiTakebe95}.
This lemma gives a restriction on the global behaviour of~$\omega^{(0)}_2$ which does not involve any choice of the point and the local coordinate.
In particular, observe that the equation $z_1=z_2$ may be satisfied for two distinct points and the divisor of poles of~$\omega^{(0)}_2$ may have components distinct from the diagonal. The only case when these additional poles do not appear is when~$z$ is of degree~$1$, that is, when~$\Sigma=\C P^1$ and~$z$ is an affine coordinate. Thus, we have the following immediate corollary:

\begin{corollary} \label{cor:no-other-poles} 
	If $\omega^{(0)}_2$ has no other poles than on the diagonal and $\{\omega^{(g)}_n\}$ satisfy the KP integrability property, then the spectral curve is rational.
\end{corollary}

This condition on $\omega^{(0)}_2$ is always satisfied, by definition, in the set-up of standard (c.f. \cite{EO}) topological recursion:

\begin{theorem}\label{thm:TR-rational-curve-refined} 
If the differentials $\{\omega^{(g)}_n\}$ produced by standard topological recursion are KP integrable, then the spectral curve $\Sigma$ is rational. Moreover, in this case the kernel~$\bK$ is defined globally on~$\Sigma^2$ and the determinantal formulas written in Equation~\eqref{eq:omega-kappa-relation} hold globally on~$\Sigma^n$.
\end{theorem}

\begin{proof}
The second statement follows from the fact that the topological recursion differentials have no residues so that the integrals entering the definition of $\bK$ under the exponent are univalued in this case.	
\end{proof}


\subsection{\texorpdfstring{$x-y$}{x-y} swap for Baker--Akhiezer kernel}
If the systems of differentials $\{\omega^{(g)}_n\}$ and hence $\{\omega^{\vee,(g)}_n\}$ are KP integrable, in the complicated combinatorial expression~\eqref{eq:MainFormulaSimple} of $x-y$ duality they both can be replaced by a simpler determinantal formulas~\eqref{eq:omega-kappa-relation}. This observation raises a natural question: what happens with the Baker--Akhiezer kernel under the $x-y$ swap? It proves out that there are explicit formulas relating $x-y$ dual kernels.

Let us consider two integral transforms which are inverse to each other.
\begin{proposition}\label{prop:identity-Gauss1} 
Let for some function $f^\vee(z)$
\begin{equation}\label{eq:etr}
f(z)=\tfrac{\ii}{\sqrt{2\pi\,\hbar}}
\int f^\vee(\chi)\;
y'(\chi)\;
e^{\frac1\hbar\left(x(z)(y(z)-y(\chi))+\int_z^{\chi}x\,dy\right)} \,d\chi.
\end{equation}
Then
\begin{equation}\label{eq:etri}
f^\vee(z)=\tfrac{1}{\sqrt{2\pi\,\hbar}}
\int f(\chi)\;
x'(\chi)\;
e^{-\frac1\hbar\left(y(z)(x(z)-x(\chi))+\int_z^{\chi}y\,dx\right)}\,d\chi.
\end{equation}
\end{proposition}
A proof of Proposition \ref{prop:identity-Gauss1} is given in Section \ref{sec:FormalGaussianIntegral}.

\begin{remark}
All integrals of this type are understood purely formally in the sense of asymptotic expansions for small absolute value of $\hbar$ near the critical points $\chi=z$. With this convention the expressions become formal Gaussian integrals and the coefficient of each power of~$\hbar$ of this integral is a finite order differential operator applied to~$f^\vee$ or $f$, respectively. We refer to Section~\ref{sec:Gaussian} for the precise definition and examples of how one can manipulate with these integrals within the framework given by this formal meaning.
\end{remark}

\begin{remark}
	We treat the functions in \eqref{eq:etr},\eqref{eq:etri} as symbolic expressions, c.f. \eqref{eq:example f}.
\end{remark}

\begin{remark}
Integral transformations considered in this section are inspired by and closely related to the Kontsevich matrix model and its generalizations \cite{Konts,KMMMZ,AdlervM,Kharchev-Marshakov,MiMoSeff96,AAH3}. Relation to matrix models
in more detail will be considered elsewhere.
\end{remark}

Proposition \ref{prop:identity-Gauss1} describes an invertible transformation for arbitrary functions $x(z)$ and $y(z)$ (there are still some restrictions imposed by Gaussian integrals, see Remark \ref{Rmk6.1} below). This type of transformation connects the Baker--Akhiezer kernels on the two sides of the $x-y$ duality.

\begin{theorem} \label{thm:transformation-KKernel} 
	The Baker--Akhiezer kernels $\bK$ and $\bK^\vee$ related by the $x-y$ swap are expressed in terms of one another by the following double  integrals:
	\begin{align}\label{eq:xyK-shortform}
		K(z_1,z_2)&=\tfrac{-\ii}{2\pi\hbar}
		\iint
		K^\vee(\chi_1,\chi_2)\;
		y'(\chi_1)y'(\chi_2)\, d\chi_1d\chi_2\times
		\\ \notag &\hskip3.0cm
		e^{-\frac{1}{\hbar}\big(x(z_2)(y(z_2)-y(\chi_1))+\int_{z_2}^{\chi_1} xdy\big)}e^{\frac{1}{\hbar}\big(x(z_1)(y(z_1)-y(\chi_2))+\int_{z_1}^{\chi_2} xdy\big)};
		\\ \label{eq:xyKvee-shortform}
	K^\vee(z_1,z_2)&=\tfrac{-\ii}{2\pi\hbar}
	\iint
	K(\chi_1,\chi_2)\;
	x'(\chi_1)x'(\chi_2) d\chi_1d\chi_2\times
	\\\notag &\hskip3.0cm
	e^{-\frac{1}{\hbar}\big(y(z_2)(x(z_2)-x(\chi_1))+\int_{z_2}^{\chi_1} ydx\big)}
	e^{\frac{1}{\hbar}\big(y(z_1)(x(z_1)-x(\chi_2))+\int_{z_1}^{\chi_2} ydx\big)}.
\end{align}
Here $\bK(z_1,z_2)=K(z_1,z_2)\sqrt{dx_1dx_2}$, $\bK^\vee(z_1,z_2)=K^\vee(z_1,z_2)\sqrt{dy_1dy_2}$, $x'=\partial_z x$, $y'=\partial_z y$.
\end{theorem}

The meaning of the involved integrals is purely formal, similarly to those of Proposition~\ref{prop:identity-Gauss1}.
We consider the integrals as asymptotic expansions in $\sqrt{\hbar}$ near the critical point, which is $\chi_1=z_2$, $\chi_2=z_1$. The contours are assumed to be chosen in such a way that the corresponding symbolic expressions are determined by the order of zeros of $(\chi_1-z_2)$ and $(\chi_2-z_1)$ near the critical point, see details in Section~\ref{sec:Gaussian}. Theorem~\ref{thm:transformation-KKernel}  is proved in Section~\ref{sec:extended-x-y-swap}.

\begin{remark} The fact that transformations given by~\eqref{eq:xyK-shortform} and~\eqref{eq:xyKvee-shortform} are inverse to each other is a direct corollary of Proposition~\ref{prop:identity-Gauss1}.
\end{remark}

\subsection{\texorpdfstring{$x-y$}{x-y} dual of a trivial case} An important special case of Theorem~\ref{thm:transformation-KKernel} concerns the situation when the $x-y$ dual side is known to be trivial.

We assume that $\omega^{\vee,(2)}_0 = \tfrac{dz_1dz_2}{(z_1-z_2)^2}$ and for all $(g,n)\not=(0,2)$ $\omega^{\vee,(g)}_n=0$. Here $z$ is a meromorphic function on $\Sigma$ that can serve as a coordinate at the point $o\in \Sigma$. In this case $Z^\vee_{o,z} = 1$ and  under the $x-y$ duality relation we have the following immediate corollary of Theorem~\ref{thm:transformation-KKernel}.

\begin{corollary} \label{cor:transformation-KKernel} The Baker--Akhiezer kernel is given by the double integral:
	\begin{align}\label{eq:xyK-shortform-cor}
		\frac{\bK(z_1,z_2)}{\sqrt{dz_1dz_2}}&=\tfrac{-\ii \sqrt{x'(z_1)x'(z_2)}}{2\pi\hbar}
		\iint
		\frac{\sqrt{y'(\chi_1)y'(\chi_2)}}{\chi_1-\chi_2}\, d\chi_1d\chi_2\times
		\\ \notag &\hskip3.0cm
		e^{-\frac{1}{\hbar}\big(x(z_2)(y(z_2)-y(\chi_1))+\int_{z_2}^{\chi_1} xdy\big)}e^{\frac{1}{\hbar}\big(x(z_1)(y(z_1)-y(\chi_2))+\int_{z_1}^{\chi_2} xdy\big)}.
\end{align}
\end{corollary}

Note that in the KP integrable case the kernel $\bK$ carries the complete information on the corresponding KP tau functions. Let $(o,z)$ be a regular point for all differentials and a local coordinate at this point, respectively. Then, the semi-infinite plane corresponding to the tau function $Z=Z_{o,z}$ is spanned by the Taylor coefficients of the expansion of $\bK(z_1,z_2)/\sqrt{dz_1dz_2}$ in~$z_2$. More explicitly, the Baker--Akhiezer kernel can be seen as a generating series for a basis  in the semi-infinite plane determining $Z$. Namely, choose an arbitrary basis $\Phi^*_{1-j}(z)$, $j=1,2,\dots$ in the space of regular power series in~$z$ and expand $\bK$ in this basis with respect to $z_2$,
\begin{equation}\label{eq:K-Phi-Phi-star}
\bK(z_1,z_2)=\sum_{i=1}^\infty \Phi_i(z_1)\Phi^*_{1-i}(z_2)\sqrt{dz_1dz_2}.
\end{equation}
Then the coefficients $\Phi_i$ of this expansion generate the corresponding semi-infinite plane as well as the corresponding decomposable vector $v=\Phi_1\wedge\Phi_2\wedge\dots$. A different basis $\Phi^*_{1-j}$ in the space of regular series leads to a different collection of vectors $\Phi_i$ but the linear span of these vectors is the same. Conversely, any basis $\Phi_i$ in the semi-infinite plane associated with the tau function can be realized by an expansion of the form~\eqref{eq:K-Phi-Phi-star}. Namely, the corresponding basis $\Phi^*_{1-i}$ in the space of regular series is orthogonal dual to the basis of functions~$\Phi_i$ in the following sense
\begin{align}\label{eq:ortho}
	\res_{z=0}\Phi_i (z) \Phi^*_j (z) dz = \delta_{i+j,1}.
\end{align}
For convenience, we will always assume that the chosen basis in the semi-infinite plane satisfies $\Phi_i (z) = z^{-i} (1+O(z))$ which is equivalent to the requirement  $\Phi^*_{1-i} (z) = z^{i-1} (1+O(z))$.


Then,  in the setting of Corollary~\ref{cor:transformation-KKernel},  and under an extra assumption of the local behavior of $x'y'$ at the point $o$, we have the following explicit formula for $\Phi_i$, $i\geq 1$.

\begin{corollary} \label{cor:Phi-TirivalDual} In the setting of Corollary~\ref{cor:transformation-KKernel}, assume $x'y'$ has a pole of degree $\geq 3$ at $z=0$. For the $x-y$ dual system $\{\omega^{(g)}_n\}$ the corresponding partition function $Z_{o,z}$ is a KP tau function corresponding to the semi-infinite plane spanned by
	\begin{align}	\label{eq:xyPhi-TrivialDual}
		\Phi_i(z)&=\sqrt{\tfrac{-x'(z)}{2\pi\hbar}}
		\int
		\chi^{-i}\;
		\sqrt{{y'(\chi)}} \;
		e^{\frac{1}{\hbar}\big(x(z)(y(z)-y(\chi))+\int_{z}^{\chi} xdy\big)}d\chi
,&i\ge 1.
	\end{align}
\end{corollary}

\begin{proof}
Indeed, rewrite an expansion of Equation~\eqref{eq:xyK-shortform-cor} as
\begin{align}\label{eq:xyK-shortform-cor-expanded}
		\frac{\bK(z_1,z_2)}{\sqrt{dz_1dz_2}}&=\tfrac{\sqrt{-x'(z_1)x'(z_2)}}{2\pi\hbar}
		\iint
		\sum_{i=1}^\infty \frac{\chi_1^{i-1}}{\chi_2^i}{\sqrt{y'(\chi_1)y'(\chi_2)}}\, d\chi_1d\chi_2\times
		\\ \notag &\hskip3.0cm
		e^{-\frac{1}{\hbar}\big(x(z_2)(y(z_2)-y(\chi_1))+\int_{z_2}^{\chi_1} xdy\big)}e^{\frac{1}{\hbar}\big(x(z_1)(y(z_1)-y(\chi_2))+\int_{z_1}^{\chi_2} xdy\big)}
		\\ \notag
&=\sum_{i=1}^\infty \Phi_i(z_1)\Phi^*_{1-i}(z_2),%
	\end{align}
where $\Phi_i$ is given by~\eqref{eq:xyPhi-TrivialDual} and
\begin{equation}
\Phi^*_{1-i}(z)=\sqrt{\tfrac{x'(z)}{2\pi\hbar}}
		\int  \chi^{i-1}\sqrt{y'(\chi)}\, d\chi
		e^{-\frac{1}{\hbar}\big(x(z)(y(z)-y(\chi))+\int_{z}^{\chi} xdy\big)}.
\end{equation}
The condition that $x'y'$ has a pole of degree $\geq  3$ at $z=0$ and the analysis of the corresponding Gaussian integrals imply that the second factor $\Phi^*_{1-i}(z_2)$ in the sum is regular in~$z_2$ and has the form $z_2^{i-1} (1+O(z_2))$. It follows that the first factor $\Phi_i(z_1)$ belongs to the corresponding semi-infinite plane and the same analysis shows that it has the form $z_1^{-j}(1+O(z_1))$ as required.
\end{proof}

\subsection{\texorpdfstring{$x-y$}{x--y} duality for dual points in the Sato--Wilson Grassmannian}

The results of Corollary~\ref{cor:Phi-TirivalDual} can be further extended without an assumption on the triviality of the dual side in terms of the so-called dual points on the Sato--Wilson Grassmannian. It is merely a different point of view on Theorem~\ref{thm:transformation-KKernel}.

Along with the kernel $\bK(z_1,z_2)$ consider also the dual kernel $\bK^*(z_1,z_2)=-\bK(z_2,z_1)$. In the expansion at a regular point~$o$ in the local coordinate~$z$ it corresponds to the dual tau function $Z^*(p_1,p_2,\dots)=Z(-p_1,-p_2,\dots)$. Expanding this kernel in $z_2$ we get
\begin{equation}\label{eq:K-star-Phi-star-Phi}
-\bK(z_2,z_1)=\bK^*(z_1,z_2)= \sum_{i=1}^\infty \Phi^*_i(z_1) \Phi_{1-i}(z_2)\sqrt{dz_1dz_2},
\end{equation}
where the functions $\Phi^*_i$, $i\ge1$ generate a basis in the semi-infinite plane corresponding to the tau function~$Z^*$ and $\Phi_{1-i}$, $i\ge1$, is the orthogonal dual basis in the space of regular power series. Remark that the choices of bases for the expansions~\eqref{eq:K-Phi-Phi-star} and~\eqref{eq:K-star-Phi-star-Phi} are independent. But for any choices we get a basis $\Phi_i$, $i\in\Z$ in the space of Laurent series in~$z$ and the orthogonal dual in the sense of~\eqref{eq:ortho} basis $\Phi_i^*$, $i\in\Z$, both satisfying $\Phi_i(z)=z^{-i}(1+O(z))$, $\Phi_i^*(z)=z^{-i}(1+O(z))$.

With this notation, the transformation rule for the Baker--Akhiezer kernel given in Theorem~\ref{thm:transformation-KKernel} has the following corollary.

\begin{corollary} \label{cor:phi-phi-phi-phi} Assume $x'y'$ has a pole of degree $\geq 3$ at $z=0$. Let $\Phi_i^{*,\vee}(z)$, $i\ge1$ have the form $z^{-i}(1+O(z))$ and generate the semi-infinite plane corresponding to the tau function $Z^{*,\vee}(p_1,p_2,\dots)=Z^{\vee}(-p_1,-p_2,\dots)$ where $Z^\vee$ is $x-y$ dual to the tau function $Z=Z_{o,z}$. Then the following vectors $\Phi_i$, $i\ge1$, have the form $\Phi_i(z)=z^{-i}(1+O(z))$ and generate the semi-infinite plane corresponding to the tau function~$Z$:
	\begin{align}	\label{eq:xy-for-Phi-dual}
		\Phi_i(z)&=\sqrt{\tfrac{-x'(z)}{2\pi\hbar}}
		\int
		\Phi_i^{*,\vee}(\chi)\;
		\sqrt{{y'(\chi)}} \;
		e^{\frac{1}{\hbar}\big(x(z)(y(z)-y(\chi))+\int_{z}^{\chi} xdy\big)}d\chi , & i\ge 1.
\end{align}
\end{corollary}


\begin{proof}
We proceed in the same way as in the proof of Corollary~\ref{cor:Phi-TirivalDual}. Using an expansion~\eqref{eq:K-star-Phi-star-Phi} for the $x-y$ dual kernel, we expand the left hand side of Equation~\eqref{eq:xyK-shortform} as
\begin{align}
	\frac{\bK(z_1,z_2)}{\sqrt{dz_1dz_2}}&=\tfrac{\sqrt{-x'(z_1)x'(z_2)}}{2\pi\hbar}
	\iint\Bigl(\sum_{i=1}^\infty \Phi^{*,\vee}_i(\chi_2)\Phi^{\vee}_{1-i}(\chi_1)\Bigr)
	y'(\chi_1)y'(\chi_2)\, d\chi_1d\chi_2\times
	\\ \notag &\hskip3.0cm
	e^{-\frac{1}{\hbar}\big(x(z_2)(y(z_2)-y(\chi_1))+\int_{z_2}^{\chi_1} xdy\big)}
    e^{\frac{1}{\hbar}\big(x(z_1)(y(z_1)-y(\chi_2))+\int_{z_1}^{\chi_2} xdy\big)}
\\\notag & = \sum_{i=1}^\infty \Phi_i(z_1)\Phi^*_{1-i}(z_2),
\end{align}
where $\Phi_i$ is given by~\eqref{eq:xy-for-Phi-dual} and 
\begin{equation}
\Phi^*_{1-i}(z)=\sqrt{\tfrac{x'(z)}{2\pi\hbar}}
		\int
		\Phi_{1-i}^{\vee}(\chi)\;
		\sqrt{{y'(\chi)}} \;
		e^{-\frac{1}{\hbar}\big(x(z)(y(z)-y(\chi))+\int_{z}^{\chi} xdy\big)}d\chi.
\end{equation}
The condition on the order of pole of $x'y'$ at $z=0$ implies that $\Phi_i$ and $\Phi^*_{1-i}$ have the form $z_1^{-i}(1+O(z_1))$ and $z_2^{i-1}(1+O(z_2))$, respectively, and hence, this expansion of the kernel~$\bK$ provides a basis in the corresponding semi-infinite plane.
\end{proof}

\begin{remark} In the context of the path integral solutions to Douglas equation one considers the so-called $p-q$ duality, which describes a relation between two minimal models coupled to two-dimensional topological gravity, see~\cite{Kharchev-Marshakov}.
In our language it should correspond to the swap of $x$ and $y$ defined to be polynomials of fixed finite degrees $p$ and $q$, respectively. We claim that in these special cases the statements of Corollary 3.17 provide the $p-q$ duality proposed in~\cite{Kharchev-Marshakov}, cf.~Equation~(27) in op. cit.
\end{remark}

\subsection{Topological recursion with non-ramified \texorpdfstring{$y$}{y}}

\label{sec:toporec-nonramified-y}

A typical source of situations covered by Corollaries~\ref{cor:Phi-TirivalDual} and \ref{cor:transformation-KKernel} comes from the theory of topological recursion. Assume $\Sigma=\C P^1$, $z$ is a global coordinate on $\C P^1$, and $y(z)=(az+b)/(cz+d)$ with $ad-bc\neq 0$. Under some conditions on $x$ and $y$ the theory of topological recursion assigns to this input a system of differentials $\{\omega^{(g)}_n\}$. These differentials are regular at the points $o\in \C P^1$ where $dx\not=0$.

\begin{theorem} \label{thm:TR-main} Let $\{\omega^{(g)}_n\}$ be a system of differentials constructed by topological recursion for the input consisting of $\Sigma=\C P^1$ with a global coordinate $z$ and $y(z)=(az+b)/(cz+d)$. Let $z=0$ be a regular point for the system of differentials. Then $Z_{0,z}$ is a KP tau function. Moreover, if $x'y'$ has a pole of degree at least $3$ at $z=0$, then the tau function $Z_{0,z}$ is described by the point of the semi-infinite Grassmannian spanned by $\Phi_i$'s given by Equation~\eqref{eq:xyPhi-TrivialDual}.
\end{theorem}

\begin{proof} For such $y$'s the $x-y$ swap relation can in fact replace topological recursion under an extra assumption that $x$ is a rational function and all zeros of $dx$ are simple. In this situation the $x-y$ swap relation actually provides us with explicit closed formulas for the system of differentials, see~\cite[Section 7]{alexandrov2022universal}.
	
	In the case when not all zeros of $dx$ are simple one has to consider the so-called Bouchard--Eynard version of topological recursion, see~\cite{B-E}, in combination with possibly having the pole of $y$ at a zero of $dx$, cf.~\cite{ChekhovNorbury}. One can embed these cases in the families of functions $x_\epsilon$ such that the given $x$ is equal to $x_0$ and for $\epsilon\not=0$ all zeros of $dx_\epsilon$ are simple.
	
	In this situation topological recursion is compatible with taking the limit $\epsilon\to 0$~\cite[Theorem 5.8]{limits}, and the $x-y$ formula is compatible with the limit $\epsilon\to 0$ by construction. This allows us to wave the assumption that all zeros of $dx$ are simple and still apply Corollary~\ref{cor:Phi-TirivalDual}.
\end{proof}

As a corollary of this theorem, we obtain the following explicit formulas in two important sets of examples (in these examples the limit procedure on the side of topological recursion is described in detail and proved in~\cite{charbonnier2022shifted,chidambaram2023relations}).

\begin{corollary}\label{cor:ex} (1) Let $x=z^{-r}/r$, $y=-z^{-1}$, $r\geq 2$. In this case the associated partition function $Z_{0,z}$ is known to be the string solution of the $r$-th Gelfand--Dickey hierarchy. The corresponding point on the semi-infinite Grassmannian is given by
\begin{equation}
			\Phi_i(z)=\tfrac{1}{\sqrt{2\pi\hbar z^{r+1}}}
	\int
	\chi^{-i-1}
	e^{\frac{1}{r\,(r+1)\,\hbar}\big(
z^{-r}-rz^{-r-1}(\chi-z)-\chi^{-r}
\big)}d\chi,\qquad i\ge1. 	
\end{equation}
(2) Let $x=z^{-r}/r$, $y=-z$, $r\geq 2$. It is the case of higher BGW discussed in Theorem~\ref{thm:r-BGW}. The corresponding point on the semi-infinite Grassmannian is given by
\begin{equation}\label{eq:hBGW}
	\Phi_i(z)=\tfrac{\ii}{\sqrt{2\pi\hbar z^{r+1}}}
	\int
	\chi^{-i}
	e^{\frac{1}{r\,(r-1)\,\hbar}\big(
-z^{1-r}+(r-1)z^{-r}(\chi-z)+\chi^{1-r}
\big)}d\chi,\qquad i\ge1. 	
\end{equation}
\end{corollary}

Theorem~\ref{thm:r-BGW} is essentially a reformulation of the part $(2)$ of Corollary~\ref{cor:ex}.
\begin{proof}[Proof of Theorem~\ref{thm:r-BGW}]
	Indeed, up to the difference of conventions, expression describing the point  \eqref{eq:hBGW}  of the semi-infinite Grassmannian
	coincides with the expression of the point of Grassmannian for the higher Br\'ezin--Gross--Witten tau function, given by \cite[Equation (4.7)]{alexandrov2023higher}. More specifically, to get the description \cite{alexandrov2023higher} one should consider the inverse parameter $z^{-1}$ instead of $z$ and then divide \eqref{eq:hBGW} by $z$.
	It remains to note that \eqref{eq:hBGW}  uniquely specify the tau function, therefore it is given by the matrix integral \eqref{eq:MMBGW}.
\end{proof}

Corollary \ref{cor:ex} can be immediately extended to the case where $x=P(z^{-1})$ with arbitrary polynomial $P$. In particular, if we take $x=z^{-r}/r-\epsilon z^{-1}$ and $y=-z$ with $r\geq 2$, then the set of the differentials constructed with
topological recursion describes intersection numbers with the deformed $r$-theta class \cite[Theorem D]{chidambaram2023relations}. We conclude that these differentials are KP integrable, moreover, the KP tau function $Z_{0,z}$ is associated with the point of the Grassmannian, generated by \eqref{eq:xyPhi-TrivialDual} with these $x$ and $y$, that is
\begin{equation}
			\Phi_i(z)=\sqrt{\tfrac{-z^{-r-1}+\epsilon z^{-2}}{2\pi\hbar}}
	\int
	\chi^{-i}
	e^{\frac{1}{\hbar}\left(\frac{1}{r\,(r-1)}\left(
-z^{1-r}+(r-1)z^{-r}(\chi-z)+\chi^{1-r}\right)+\epsilon\left(1-\chi/z+\log\left(\chi/z\right)\right)
\right)}d\chi. 	
\end{equation}

Note that this deformation coincides with the generalized higher BGW tau function, considered in \cite[Section 5]{alexandrov2023higher}. Indeed, if we consider it as a deformation of the spectral curve from Theorem \ref{thm:r-BGW}, namely $x=z^r/r-\epsilon z$, $y=z^{-1}$, then this
spectral curve can be described as $r y^{r-1}(xy+\epsilon)=1$, which coincides with the spectral curve suggested in \cite[Section 5]{alexandrov2023higher} for the generalized higher BGW tau function if one identifies $r=m+1$ and $\epsilon=S$. Moreover, the tau function $Z_{\infty,\eta^{-1}}$ for $\eta=(z^r-\epsilon r z)^{1/r}$ coincides, up to a $(0,1)$ contribution, with the generalized higher BGW tau function. The proof is based on the identification of the points of Grassmannian and is a direct generalization of proof of Theorem \ref{thm:r-BGW} where one should carefully treat the change of local coordinate.

\begin{remark}If $y(z)$ is linear fractional, then the triviality of the dual topological recursion follows from the fact that $dy$ has no zeros. However, there are other cases when the dual topological recursion is trivial and Corollary~\ref{cor:Phi-TirivalDual} is applied. It might happen that $dy$ gets zero at some point but $dx$ has a high order pole at the same point so that this point does not contribute to the poles of the dual topological recursion. For example, for the spectral curve
\begin{equation}
x=z^r,\quad y=z^{-s},\quad r>s>0,
\end{equation}
the Bouchard--Eynard version of topological recursion often makes no sense. But one can consider the $x-y$ swap relation as a properly replacement for topological recursion, which produces KP integrable differentials since the $x-y$ dual differentials are trivial. 
These tau functions should describe interesting enumerative geometry invariants, cf.~\cite{chidambaram2023relations} and~\cite[Section~7]{borot2022higher}.
\end{remark}

\begin{remark}
	It would be interesting to compare the second statement of Corollary~\ref{cor:ex} (and its immediate extensions for arbitrary polynomials $x=P(z^{-1})$) with the recently announced results of Shuai Guo, Ce Ji, and Qingsheng Zhang. In their talk ``On a generalization of Witten's conjecture through spectral curves'' at the conference ``Integrable systems and their applications" hold in Sochi, September 11-15, 2023, Shuai Guo asserted that in these cases they proved constrained KP integrability of the corresponding tau functions. At the moment of the present revision a preliminary version of their paper is available at~\cite{guo2023generalization}.
\end{remark}

\section{Introduction, part 3: the structure of the \texorpdfstring{$x-y$}{x-y} swap relation} \label{sec:intro3FullPicture}

The goal of this section is to go deeper into the structure of $x-y$ swap relation between two systems of differentials, $\{\omega^{(g)}_n\}$ and $\{\omega^{\vee,(g)}_n\}$, without any reference to integrability (and only at the very end of this section we explain what its results mean in connection to KP integrability).  In particular, we present in this section an (extended) reformulation of the $x-y$ swap relation which does not involve the complicated combinatorics of the sums over graphs of the original formulation.

The only assumptions that we have at the moment are the ones we made in Section~\ref{sec:n-pt-diff}.

\subsection{Extra objects of \texorpdfstring{$x-y$}{x--y} swap relation}

We associate with two systems of differentials and extra functions $x$ and $y$ an extended system of functions related to the spectral curve that we assemble into the following diagram:

\begin{equation}\label{eq:XYdiagram}
\begin{aligned}
{\ }   \\
\xymatrix@C=30pt{
	\{\omega^{(g)}_n\} \ar @ <0.5ex> [r] \ar @/^2pc/ [rr] & \{\bW^{(g)}_n\} \ar @<0.5ex> [l] \ar [ld] \ar  @<0.5ex> [d] & \{\Omega_n\} \ar[l] \ar @<0.5ex> [d] \\
	\{\omega^{\vee,(g)}_n\}\ar @ <0.5ex> [r] \ar @/_2pc/ [rr] & \{\bW^{\vee,(g)}_n\} \ar @<0.5ex> [l]  \ar [lu] \ar @<0.5ex> [u]& \{\Omega^{\vee}_n\} \ar[l] \ar @<0.5ex> [u]
}
\\
{\ }
\end{aligned}
\end{equation}

The arrows in this diagram are certain transformations defined by explicit relations presented below and producing the object at the head of the arrow from the corresponding object on the tail. The main statement of this section claims that this diagram is commutative. It means that we can use transformations represented by some of the arrows as the definitions of the corresponding objects and then the remaining arrows represent universal functional identities.
Remark that the identities implied by the left square in this diagram are actually treated in~\cite{alexandrov2022universal} under slightly different notations, while the identities implied by the right square of the diagram is an essential novelty of the present paper.
The definitions of all involved objects is going to be ad hoc initially, but then they will be justified by the way they clarify the internal structure of Equation~\eqref{eq:MainFormulaSimple} and by the role they play in the KP integrable case.
So let us introduce all objects and arrows in this diagram.

First of all we have:
\begin{multline} \label{eq:def-bW}
	 \bW^{(g)}_n = \bW^{(g)}_n(z_{\llbracket n \rrbracket},u_{\llbracket n\rrbracket}) \coloneqq
	 \coeff \hbar {2g}\prod_{i=1}^n \frac{e^{- u_i(\cS(u_i\hbar\partial_{x_i})-1) y_i}}{u_i}
 \sum_{\Gamma} \frac{\hbar^{2g(\Gamma)}}{|\mathrm{Aut}(\Gamma)|}
\\
\prod_{e\in E(\Gamma)}
\prod_{j=1}^{|e|}\restr{(\tilde u_j, \tilde x_j)}{ (u_{e(j)},x_{e(j)})} \tilde u_j \cS(\tilde u_j \hbar  \partial_{\tilde x_j}) \sum_{\tilde g=0}^\infty \frac{\hbar^{2\tilde g}\tilde \omega^{(\tilde g)}_{|e|}(\tilde z_{\llbracket |e|\rrbracket})}  {\prod_{j=1}^{|e|} {d\tilde x_j}}.
\end{multline}
Here as usual $x_i = x(z_i)$, $\tilde x_i = x(\tilde z_i)$, and $y_i=y(z_i)$. The sum is taken over all connected graphs $\Gamma$ with $n$ labeled vertices and multiedges of index $\geq 1$ with the same reservations as in~\eqref{eq:MainFormulaSimple}. In particular, we exclude the contribution of $(0,1)$ edges in the graphs --- this contribution is accounted in the first factor of the formula and it corresponds to the convention $\omega^{(0)}_1=-y\,dx$. The functions $\bW^{(g)}_n$ are symmetric functions in $n$ pairs of variables $(z_i,u_i)$. We regard $z_i$ as a point of the $i$th factor in~$\Sigma^n$ while $u_i$ is just a formal parameter. With the exception of $\bW^{(0)}_1=\frac{1}{u_1}$  corresponding to leading term of the trivial graph, $\bW^{(g)}_n$ is a polynomial in $u_1,\dots,u_n$ whose coefficients are meromorphic functions on $\Sigma^n$ represented as finite symbolic expressions in $\omega^{(g)}_n$'s, $x'$ and $y'$.

We use exactly the same definition in the $x-y$ dual case, with the roles of $x$ and $y$ interchanged:
\begin{multline} \label{eq:def-bW-vee}
	\bW^{\vee,(g)}_n = \bW^{\vee,(g)}_n(z_{\llbracket n \rrbracket},v_{\llbracket n\rrbracket}) \coloneqq
	\coeff \hbar {2g}
\prod_{i=1}^n \frac{e^{-v_i (\cS(v_i\hbar \partial_{y_i})-1)x_i}}{v_i}
	 \sum_{\Gamma} \frac{\hbar^{2g(\Gamma)}}{|\mathrm{Aut}(\Gamma)|}
   \\
	\prod_{e\in E(\Gamma)}
		 \prod_{j=1}^{|e|}\restr{(\tilde v_j, \tilde y_j)}{ (v_{e(j)},y_{e(j)})} \tilde v_j \cS(\tilde \hbar v_j  \partial_{\tilde y_j}) \sum_{\tilde g=0}^\infty \frac{\hbar^{2\tilde g}\tilde \omega^{\vee,(\tilde g)}_{|e|}(\tilde z_{\llbracket |e|\rrbracket})}  {\prod_{j=1}^{|e|} {d\tilde y_j}}.
\end{multline}

In a similar way, we define
\begin{multline} \label{eq:def-Omega}
	\Omega_n = \Omega_n(w_{\llbracket n \rrbracket},\bar w_{\llbracket n\rrbracket}) \coloneqq
	 \prod_{i=1}^n \frac{1}{x_i-\bar x_i}
	\sum_{\Gamma} \frac{1}{|\mathrm{Aut}(\Gamma)|}
	\\
    \prod_{e\in E(\Gamma)}
		\prod_{j=1}^{|e|}\restr{(\tilde w_j, \tilde {\bar w}_j)}{ (w_{e(j)},\bar w_{e(j)})}
	\int_{\tilde{\bar w}_1}^{\tilde w_1}\!\!\!\int_{\tilde{ \bar w}_2}^{\tilde w_2}\!\!\!\!\cdots \!\!\int_{\tilde {\bar w}_{|e|}}^{\tilde w_{|e|}} \sum_{\tilde g=0}^\infty \hbar^{2\tilde g- 2+|e|} \tilde \omega^{(\tilde g)}_{|e|}(\tilde z_{\llbracket |e|\rrbracket})
\end{multline}
and
\begin{multline} \label{eq:def-Omega-vee}
	\Omega^{\vee}_n = \Omega^{\vee}_n(w_{\llbracket n \rrbracket},\bar w_{\llbracket n\rrbracket}) \coloneqq
	 \prod_{i=1}^n \frac{1}{y_i-\bar y_i}
	 \sum_{\Gamma} \frac{1}{|\mathrm{Aut}(\Gamma)|}
   \\
	\prod_{e\in E(\Gamma)}
	\prod_{j=1}^{|e|}\restr{(\tilde w_j, \tilde {\bar w}_j)}{ (w_{e(j)},\bar w_{e(j)})}
	\int_{\tilde{\bar w}_1}^{\tilde w_1}\!\!\!\int_{\tilde{ \bar w}_2}^{\tilde w_2}\!\!\!\!\cdots \!\!\int_{\tilde {\bar w}_{|e|}}^{\tilde w_{|e|}} \sum_{\tilde g=0}^\infty \hbar^{2\tilde g-2 + |e|} \tilde\omega^{\vee,(\tilde g)}_{|e|}(\tilde z_{\llbracket |e|\rrbracket}).
\end{multline}
These are the functions on $\Sigma^{2n}$ defined in the neighborhoods of the product of diagonals in $\Sigma^2$ corresponding to the pairs of variables $(w_i,\bar w_i)$.

Comparing the definitions of $\bW^{(g)}_n$ and $\Omega_n$ we see that $\bW^{(g)}_n$ can be obtained from $\Omega_n$ by the substitution
$w_i=e^{\frac{\hbar u_i}2 \partial_{x_i}}z_i$, $\bar w_i=e^{-\frac{\hbar u_i}2 \partial_{x_i}}z_i$ with an additional account of $(0,1)$ contribution. Indeed, using that
\begin{equation}
\int\limits_{e^{-\frac{\hbar u}2 \partial_{x}}z}^{e^{\frac{\hbar u}2 \partial_{x}}z}=u\hbar \cS(u\hbar\partial_x)\frac{1}{dx},
\qquad e^{\pm\frac{\hbar u}2 \partial_{x}}x(z)=x(z)\pm\tfrac{u\hbar}{2},
\end{equation}
we obtain
	\begin{equation} \label{eq:Omega-to-bW}
		\bW_n^{(g)} (z_{\llbracket n \rrbracket})=\coeff\hbar{2g-2+n}  \prod_{i=1}^n e^{- u_i(\cS(u_i\hbar\partial_{x_i})-1) y_i} \restr{(w_i,\bar w_i)}{\Big(e^{\frac{\hbar u_i}2 \partial_{x_i}}z_i,e^{-\frac{\hbar u_i}2 \partial_{x_i}}z_i\Big)} \Omega_n (w_{\llbracket n \rrbracket},\bar w_{\llbracket n\rrbracket}),
\end{equation}
and, similarly,
\begin{equation}  \label{eq:Omega-to-bW-vee}
				\bW^{\vee,g}_n (z_{\llbracket n \rrbracket})=\coeff\hbar{2g-2+n}  \prod_{i=1}^n  e^{- v_i(\cS(v_i\hbar\partial_{y_i})-1) x_i}  \restr{(w_i,\bar w_i)}{\Big(e^{\frac{\hbar v_i}2 \partial_{y_i}}z_i,e^{-\frac{\hbar v_i}2 \partial_{y_i}}z_i\Big)} \Omega^\vee_n (w_{\llbracket n \rrbracket},\bar w_{\llbracket n\rrbracket}).
	\end{equation}

\begin{remark}
The substitution $w_i=e^{\frac{\hbar u_i}2 \partial_{x_i}}z_i$, $\bar w_i=e^{-\frac{\hbar u_i}2 \partial_{x_i}}z_i$ means that the arguments $(z_i,u_i)$ of the function $\bW^{(g)}_n$ represent a point in a neighborhood of the diagonal in~$\Sigma^2$, and the substitution provides a parametrization of this neighborhood.
\end{remark}

Notice that for $2g-2+n\geq 0$ we have
\begin{align}
	  \label{eq:omega-vee-from-Omega}
	\frac{\omega^{(g)}_n (z_{\llbracket n \rrbracket})}{\prod_{i=1}^n dx_i}  & = \restr{u_{\llbracket n \rrbracket}}{0}\bW^{(g)}_n (z_{\llbracket n \rrbracket}, u_{\llbracket n \rrbracket})
=\coeff\hbar{2g-2+n}\restr{w_{\llbracket n \rrbracket},\bar w_{\llbracket n \rrbracket}}{z_{\llbracket n \rrbracket}}\Omega_n;
\\ 	  \label{eq:omega-from-Omega}
\frac{\omega^{\vee,(g)}_n (z_{\llbracket n \rrbracket})}{\prod_{i=1}^n dy_i}  & = \restr{v_{\llbracket n \rrbracket}}{0}\bW^{\vee, (g)}_n (z_{\llbracket n \rrbracket}, v_{\llbracket n \rrbracket})
=\coeff\hbar{2g-2+n}\restr{w_{\llbracket n \rrbracket},\bar w_{\llbracket n \rrbracket}}{z_{\llbracket n \rrbracket}}\Omega^\vee_n.
\end{align}
Also notice that the relation of $x-y$ duality~\eqref{eq:MainFormulaSimple} and its dual counterpart take with the introduced notations the following form
\begin{align}
  \label{eq:omega-vee-from-bW}
	\frac{\omega^{\vee,(g)}_n(z_{\llbracket n \rrbracket}) }{\prod_{i=1}^n dy_i} & = (-1)^n\prod_{i=1}^n \bigg(\sum_{r_i=0}^\infty \partial_{y_i}^{r_i} [u_i^{r_i}] \frac{dx_i}{dy_i} \bigg) \bW^{(g)}_n (z_{\llbracket n \rrbracket}, u_{\llbracket n \rrbracket});
\\ \label{eq:omega-from-bW-vee}
	\frac{\omega^{(g)}_n (z_{\llbracket n \rrbracket})}{\prod_{i=1}^n dx_i} & = (-1)^n \prod_{i=1}^n \bigg(\sum_{r_i=0}^\infty \partial_{x_i}^{r_i} [v_i^{r_i}] \frac{dy_i}{dx_i}  \bigg) \bW^{\vee,(g)}_n (z_{\llbracket n \rrbracket}, v_{\llbracket n \rrbracket}).
\end{align}

All this together represents the following part of Diagram~\eqref{eq:XYdiagram}.
\begin{equation} \label{eq:XY-diag-triv}
\begin{aligned}
~\\
\xymatrix@C=30pt{
	\{\omega^{(g)}_n\} \ar @ <0.5ex> [r] \ar @/^2pc/ [rr] & \{\bW^{(g)}_n\} \ar @<0.5ex> [l] \ar [ld]  & \{\Omega_n\} \ar[l]  \\
	\{\omega^{\vee,(g)}_n\}\ar @ <0.5ex> [r] \ar @/_2pc/ [rr] & \{\bW^{\vee,(g)}_n\} \ar @<0.5ex> [l]  \ar [lu] & \{\Omega^{\vee}_n\} \ar[l]
}
\\~
\end{aligned}
\end{equation}

The missing arrows are described in the next section.

\begin{remark}
Relations~\eqref{eq:omega-vee-from-Omega} and~\eqref{eq:omega-from-Omega} suggest that a more invariant version of the functions $\bW^{(g)}_n$, $\bW^{\vee,(g)}_n$,  $\Omega_n$,  $\Omega^\vee_n$ would be the differentials $\bW^{(g)}_n\prod_{i=1}^n dx_i$, $\bW^{\vee,(g)}_n\prod_{i=1}^n dy_i$ and half-differentials $\Omega_n\prod_{i=1}^n\sqrt{dx_id\bar x_i}$,  $\Omega^\vee_n\prod_{i=1}^n\sqrt{dy_id\bar y_i}$, respectively. However, in order to keep the exposition more explicit we prefer to deal with functions rather than with (half)differentials.
\end{remark}

\subsection{Extended \texorpdfstring{$x-y$}{x-y} swap relations}

\label{sec:extended-x-y-swap}

It proves out that there are equations relating the functions $\bW^{g}_n$ and $\bW^{\vee,(g)}_n$, or, respectively, $\Omega_n$ and $\Omega^\vee_n$ directly, avoiding complicated combinatorics of summation over graphs. These equations form the missing part of Diagram~\eqref{eq:XYdiagram}:
\begin{align}
	\xymatrix@C=30pt{
		\{\omega^{(g)}_n\}  & \{\bW^{(g)}_n\} \ar @<0.5ex> [d]  & \{\Omega_n\}  \ar @<0.5ex> [d] \\
		\{\omega^{\vee,(g)}_n\} & \{\bW^{\vee,(g)}_n\}\ar @<0.5ex> [u] & \{\Omega^{\vee}_n\}  \ar @<0.5ex> [u]
	}
\end{align}

\begin{theorem}\label{thm:transformation-Omega}
The $x-y$ swap relations admit the following extensions:
	\begin{align}
\label{eq:bW-from-bW-vee}
	\bW^{(g)}_n (z_{\llbracket n \rrbracket}, u_{\llbracket n \rrbracket}) & = (-1)^n\prod_{i=1}^n \bigg(e^{u_i y_i }\sum_{r_i=0}^\infty \partial_{x_i}^{r_i} [v_i^{r_i}] \frac{dy_i}{dx_i} e^{ -u_i y_i} \bigg) \bW^{\vee,(g)}_n (z_{\llbracket n \rrbracket}, v_{\llbracket n \rrbracket}),
	\\ \label{eq:bW-vee-from-bW}
	\bW^{\vee,(g)}_n (z_{\llbracket n \rrbracket}, v_{\llbracket n \rrbracket}) & =  (-1)^n\prod_{i=1}^n \bigg(e^{v_i x_i }\sum_{r_i=0}^\infty \partial_{y_i}^{r_i} [u_i^{r_i}] \frac{dx_i}{dy_i} e^{ -v_i x_i} \bigg) \bW^{(g)}_n (z_{\llbracket n \rrbracket}, u_{\llbracket n \rrbracket}),
\\ \label{eq:Omega-int}
	\Omega_n(w_{\llbracket n \rrbracket},\bar w_{\llbracket n \rrbracket})
&= \Bigl(\frac{-\ii}{2\pi\hbar}\Bigr)^n
	\int\!\!\cdots\!\!\!\int
	\Omega^\vee_n(\chi_{\llbracket n \rrbracket},\bar \chi_{\llbracket n \rrbracket})
	\prod_{i=1}^n y'(\chi_i)y'(\bar \chi_i)  \prod_{i=1}^n  d\chi_id\bar \chi_i\times
	\\ \notag &\qquad
	\prod_{i=1}^n e^{-\frac{1}{\hbar}\big(x(\bar w_i))(y(\bar w_i)-y(\chi_i))+\int_{\bar w_i}^{\chi_i} xdy\big)}
                  e^{\frac{1}{\hbar}\big(x(w_i)(y(w_i)-y(\bar \chi_i))+\int_{w_i}^{\bar \chi_i} xdy\big)},
	\\
	\label{eq:Omega-vee-int}
	\Omega^\vee_n(w_{\llbracket n \rrbracket},\bar w_{\llbracket n \rrbracket})
&= \Bigl(\frac{-\ii}{2\pi\hbar}\Bigr)^n
\int\!\!\cdots\!\!\!\int
\Omega_n(\chi_{\llbracket n \rrbracket},\bar \chi_{\llbracket n \rrbracket})
\prod_{i=1}^n {x'(\chi_i)x'(\bar \chi_i)}  \prod_{i=1}^n  d\chi_id\bar \chi_i\times
\\ \notag &\qquad
\prod_{i=1}^n e^{-\frac{1}{\hbar}\big(y(\bar w_i))(x(\bar w_i)-x(\chi_i))+\int_{\bar w_i}^{\chi_i} ydx\big)}
              e^{\frac{1}{\hbar}\big(y(w_i)(x(w_i)-x(\bar \chi_i))+\int_{w_i}^{\bar \chi_i} ydx\big)},
\end{align}
where the formal Gaussian integrals are understood in the same way as in Theorem~\ref{thm:transformation-KKernel}.
\end{theorem}

The first two relations of this theorem concerning the functions $\bW^{(g)}_n$ and $\bW^{\vee,(g)}_n$ are proved in a bit different formulation in~\cite{alexandrov2022universal}.
An extended reference to the op.~cit. and a proof of the other two relations is given in Section~\ref{sec:ReformulationXY-duality} below.

\begin{proof}[Proof of Theorem~\ref{thm:transformation-KKernel}] This theorem is a direct corollary of Theorem~\ref{thm:transformation-Omega}. Indeed, by construction, $\Omega_1(w,\bar w) = K(w,\bar w)$ and $\Omega^\vee_1(w,\bar w) = K^\vee(w,\bar w)$.
\end{proof}

\subsection{Extended determinantal formulas}\label{Sec:proofof2.5}

We want to stress that in the whole Section~\ref{sec:intro3FullPicture} up to this point and in particular in Theorem~\ref{thm:transformation-Omega} we don't use any assumption on integrability. Meanwhile, in the KP integrable case some of the relations admit simplification. Namely, we have the following strengthening of Theorem~\ref{thm:KP-duality}.

\begin{theorem}\label{thm:KPintegr-extended}
Assume that the system of differentials $\{\omega^{(g)}_n\}$ is KP integrable in the sense of Theorem~\ref{thm:KPglobal}. Then the functions $\Omega_n$ and the $x-y$ dual functions $\Omega^\vee_n$ admit determinantal presentations
\begin{align}\label{eq:Omega-Kdet}
\Omega_n(w_{\llbracket n \rrbracket},\bar w_{\llbracket n \rrbracket})
 &=(-1)^{n-1} \sum_{\sigma\in C_n} \prod_{i=1}^n K(w_{i},\bar w_{\sigma(i)}), & n\geq 1,
\\\label{eq:Omega-vee-Kdet}
\Omega^\vee_n(w_{\llbracket n \rrbracket},\bar w_{\llbracket n \rrbracket})
 &=(-1)^{n-1} \sum_{\sigma\in C_n} \prod_{i=1}^n K^\vee(w_{i},\bar w_{\sigma(i)}), & n\geq 1,
\end{align}
where $K(w,\bar w)=\frac{\bK(w,\bar w)}{\sqrt{dx(w)\,dx(\bar w)}}$ and $K^\vee(w,\bar w)=\frac{\bK^\vee(w,\bar w)}{\sqrt{dy(w)\,dy(\bar w)}}$ are the same as in Theorem~\ref{thm:transformation-KKernel}.
\end{theorem}

\begin{proof}
For the first equality we should use computations in the Fock space and apply Wick formula, the same as for the proof of~\eqref{eq:W-K-relation}. To this end, we just observe that in a local expansion in the KP integrable case we can compute the connected correlation functions associated with the following expression
\begin{align}\label{eq:Omega-fermion}
	\langle 0| \prod_{i=1}^n \bigg(\sum_{l_i\in \Z} \psi_{l_i} \bar w_i^{-l_i-1}  \sum_{r_i\in \Z} \psi^*_{r_i} w_i^{r_i}\bigg) v
\end{align}
(here $v$ is the corresponding decomposable vector as in~\eqref{eq:decomposable-vector}) in two different ways. The first way is to rewrite each factor $\sum_{l_i\in \Z} \psi_{l_i} \bar w_i^{-l_i-1}  \sum_{r_i\in \Z} \psi^*_{r_i} w_i^{r_i}$ as a bosonic operator and use the argument in~\cite[Section 3 and 4]{BDKS-first}, which gives the left hand side of~\eqref{eq:Omega-Kdet}. On the other hand, one can compute~\eqref{eq:Omega-fermion} using the Wick formula, and this gives the right hand side of~\eqref{eq:Omega-Kdet}.

Now to prove \eqref{eq:Omega-vee-Kdet} we will use Theorem \ref{thm:transformation-Omega}. Substitute the right hand side of \eqref{eq:Omega-Kdet} to \eqref{eq:Omega-vee-int}. Then the determinant expression in the right hand side of \eqref{eq:Omega-Kdet} allows to apply the $x-y$ duality transformation for $K$ from Theorem \ref{thm:transformation-KKernel}, and finally we obtain exactly \eqref{eq:Omega-vee-Kdet}.

\end{proof}

Now we are ready to prove one of the main results of the paper, using Theorem \ref{thm:transformation-Omega}, which we prove in Section \ref{sec:ReformulationXY-duality}.
\begin{proof}[Proof of Theorem~\ref{thm:KP-duality}]
If $\{\omega^{(g)}_n\}$ is KP integrable then we have~\eqref{eq:Omega-Kdet} and hence~\eqref{eq:Omega-vee-Kdet}. Restricting to the product of diagonals $w_i=\bar w_i$ we establish the validity of the determinantal formulas for the $x-y$ dual differentials $\{\omega^{\vee(g)}_n\}$, which implies, in turn, their KP integrability.
\end{proof}


\part{Details and proofs}

\section{Semi-classical limit equation and implications}\label{sec:quasiclassical}

In this Section we prove Proposition~\ref{prop:omega02KP} in a bit more technical but convenient formulation. Recall that the KP integrability property for a system of differentials can be reformulated as a system of equations~\eqref{eq:omega-kappa-relation} and~\eqref{eq:bK-def} on these differentials. In particular, let us consider the leading term in~$\hbar$ of~\eqref{eq:omega-kappa-relation} for $n=2$:
\begin{align}\label{eq:omega02eq}
	{\omega^{(0)}_2}=e^{ \int\limits_{x_2}^{x_1}\int\limits_{x_2}^{x_1}
		(\omega^{(0)}_2-\omega^{(0)\rm sing}_2)}
	{ \omega^{(0)\rm sing}_2},
\end{align}
where
\begin{equation}
	\omega^{(0)\rm sing}_2=\frac{dx_1dx_2}{(x_1-x_2)^2}.
\end{equation}
This equality can be considered as an equation on the unknown bidifferential~$\omega^{(0)}_2$. Proposition~\ref{prop:omega02KP} is obtained by solving this equation. Namely, we reformulate Proposition~\ref{prop:omega02KP} in the following more explicit form.

\begin{lemma}\label{lem:omega02KP} 
	 Consider a bi-differential $\omega^{(0)}_2$ on $\Sigma^2$, which has a double pole on the diagonal with bi-residue $1$. Assume that $o\in \Sigma$ is a regular point for $\omega^{(0)}_2$. Let $x$ be a local coordinate at $o$. Then all solutions of Equation~\eqref{eq:omega02eq} are given by
	\begin{equation}\label{eq:omega02z2}
		\omega^{(0)}_2=\frac{dz_1dz_2}{(z_1-z_2)^2},\quad z_i=z(x_i),
	\end{equation}
	where $z(x)$ is certain Laurent series at $o$ that extends as a global meromorphic function~$z$ on $\Sigma$.
	
	For a given $\omega^{(0)}_2$ that resolves Equation~\eqref{eq:omega02eq}
	the corresponding function~$z$ is defined uniquely up to a linear fractional transformation.
\end{lemma}

\begin{proof}
	
At the first step of the proof we are going to solve this equation in the space of formal symmetric power expansions in~$x_1$ and~$x_2$ (as it is done in~\cite[Proposition B.2]{TakasakiTakebe95}). In fact, it follows from~\eqref{eq:omega02eq-solution} that the equation of Lemma is independent of the choice of local coordinate~$x$. This implies immediately  that any bidifferential of the form~\eqref{eq:omega02z2} does satisfy this equation. So, for the first step of the argument it is sufficient to show that~\eqref{eq:omega02eq} has no solutions other than those of the form~\eqref{eq:omega02z2}.

Moreover, since Equation~\eqref{eq:omega02eq} is invariant with respect to a change of local coordinate, we can choose the local coordinate in any convenient way such that the given solution~$\omega^{(0)}_2$ takes as simple form as possible. Denote by $\alpha$ the $1$-form obtained from $\omega^{(0)}_2$ by specifying any nonzero tangent vector at  the point $o\in\Sigma$ for the second argument in $\omega^{(0)}_2$. Then $\alpha$ has a pole of order~$2$ at~$o$ with the trivial residue, and we can choose a local coordinate~$z$ such that $\alpha=\frac{dz}{z^2}$, that is,
\begin{equation}\label{eq:choice-loc-z}
	z=-\Big(\int \alpha\Big)^{-1},	
\end{equation}
with an arbitrary choice of the integration constant. Then, for the expansion
\begin{equation}\label{eq:H02z1}
	\omega^{(0)}_2=\frac{dz_1dz_2}{(z_1-z_2)^2}+d_1d_2H^{(0)}_2,\qquad H^{(0)}_2=\sum_{i,j=1}^\infty c_{i,j}z_1^iz_2^j,
\end{equation}
our choice of the local coordinate~$z$ is equivalent to requirement $c_{1,i}=c_{i,1}=0$ for all~$i\ge 1$.

Let us solve~\eqref{eq:omega02eq} inductively by computing homogeneous terms of $H^{(0)}_2$ in~\eqref{eq:H02z1} step by step. Denote by $h_d(z_1,z_2)=\sum_{i=1}^{d-1}c_{i,d-i}z_1^iz_2^{d-i}$ the degree~$d$ term of $H^{(0)}_2$. Then for $h_d(z_1,z_2)$ we obtain an equation of the following kind
\begin{equation} \label{eq:inductionH02}
	\frac{\partial^2h_d(z_1,z_2)}{\partial z_1\partial z_2}-\frac{h_d(z_1,z_1)-h_d(z_1,z_2)-h_d(z_2,z_1)+h_d(z_2,z_2)}{(z_1-z_2)^2}=\text{r.h.s}.
\end{equation}
where the right hand side is expressed in terms of the polynomials $h_{d'}$ for $d'<d$ computed in the previous steps of the induction.

This equation implies a system of linear equations for $c_{i,j}$, $i,j\geq 1$, $i+j=d$. We are only interested in the solutions of this system of equations that satisfy $c_{1,d-1}=c_{d-1,1}=0$ (as it is demanded by our choice of the coordinate $z$) and is symmetric under the swap of indices, $c_{i,j}=c_{j,i}$.
In order to prove that there are no non-trivial solutions of this linear system that satisfy these properties let us consider the left hand side of~\eqref{eq:inductionH02} as an operator $A_d$ acting in the vector space of degree $d-2$ homogeneous symmetric polynomials in $z_1,z_2$ and taking $z_1^{-1}z_2^{-1}h_d(z_1,z_2)$ to the left hand side of~\eqref{eq:inductionH02}. It is easy to find explicitly the eigenbasis of~$A_d$. Namely, it is checked by direct substitution that the polynomial $v_{d,0}=\frac{z_1^{d-1}-z_2^{d-1}}{z_1-z_2}=\sum_{i=0}^{d-2}z_1^iz_2^{d-2-i}$ is an eigenvector of~$A_d$ with the eigenvalue~$0$, and the following polynomial
\begin{equation}
	v_{d,i}=(z_1z_2)^{i-1}\Bigl(\frac{z_1^j-z_2^j}{z_1-z_2}-\frac{j}{2}(z_1^{j-1}+z_2^{j-1})\Bigr),
	\qquad j=d+1-2i,
\end{equation}
is an eigenvector of~$A_d$ with the eigenvalue $i(d-i)$ for $i=1,\dots,\lfloor\frac{d}{2}\rfloor$.
It follows that $A_d$ has one-dimensional kernel spanned by $v_{d,0}$ and it is non-degenerate on the span of the remaining eigenvectors consisting of symmetric polynomials vanishing on the diagonal $z_1=z_2$.

Now we can return to the analysis of Equation~\eqref{eq:inductionH02}. If we want to use it to determine $h_d(z_1,z_2)$ assuming that $h_{d'}\equiv 0$ for $d'<d$ (which is trivially true for $d'<2$), then we see from the above analysis of the structure of the operator $A_d$ that the only non-trivial solution that we get is proportional to $z_1z_2v_{d,0}=\sum_{i=1}^{d-1} z_1^i z_2^{d-i}$. Since by our choice of the coordinate $z$ we have an additional assumption that the coefficients of $z_1 z_2^{d-1}$ and $z_1^{d-1}z_2$ must be equal to zero, we conclude that  $h_d(z_1,z_2)$ must be equal to zero as well. This holds for any $d\geq 2$.
So, we conclude $H^{(0)}_2=0$ and $\omega^{(0)}_2=\frac{dz_1dz_2}{(z_1-z_2)}$ for the specified choice of the local coordinate, as claimed.

\begin{remark}\label{rem:qterms} Note that though we insisted above on a convenient choice of a coordinate $z$ near $o$, we could also just derive the same change of variables perturbatively trivializing the coefficients of the zero eigenvectors of $A_d$, $d\geq 2$, in any coordinate. To this end just note that
	\begin{equation}
		\frac{d(z_1 + \epsilon z_1^d) d(z_2+\epsilon z_2^d)}{(z_1 + \epsilon z_1^d-(z_2+ \epsilon z_2^d))^2} = \frac{dz_1dz_2}{(z_1-z_2)^2}+\epsilon d_1d_2 \sum_{i=1}^{d-1} z_1^i z_2^{d-i} + O(\epsilon).
	\end{equation}
	Note, however, in our approach we defined $z$ as an actual local coordinate near $o\in \Sigma$, not just as a formal power series in $x$, so we obtained a stronger statement than we initially intended to do at the first step of the proof.
\end{remark}

So, we obtained that the equality~\eqref{eq:omega02z2} holds in a suitably defined local coordinate $z$ in a neighborhood of the point~$(o,o)\in\Sigma^2$. At the second step of the proof we want to show that the introduced local coordinate~$z$ extends as a global meromorphic function, which would imply that~\eqref{eq:omega02z2} extends as an equality of globally defined meromorphic bidifferentials on~$\Sigma$.

By construction, there exists an open neighborhood $o\in U\subseteq \Sigma$ such that $z$ is a local coordinate on $U$ and
\begin{equation} \label{eq:w02-normalform-onU}
	\omega^{(0)}_2\vert_{U\times U} = \frac{dz_1dz_2}{(z_1-z_2)^2}.
\end{equation}
The local coordinate $z$ is defined up to a fractional linear transformation $z\to az/(cz+d)$, $a,c,d\in \mathbb{C}$, $ad\not=0$ (it is the effect of the choices of tangent vector and integration constant hidden in~\eqref{eq:choice-loc-z}). From the fact that $\omega^{(0)}_2$ is globally defined it follows that $dz/z^2$ extends to a global meromorphic form on $\Sigma$.

Applying the same construction of a convenient local coordinate to any other point $o'\in U$ and using  Equation~\eqref{eq:w02-normalform-onU} for $\omega^{(0)}_2$ on $U\times U$, we obtain that for each fractional linear transformation $w=(az +b) / (cz+d)$, $a,b,c,d\in \mathbb{C}$, $ad-bc\not=0$, $w(o')=0 \Leftrightarrow z(o')=-b/a$, that is, $-b/a$ is sufficiently close to $0$, we still have that $dw/w^2$ extends to a global meromorphic form on $\Sigma$.

So, not only $dz/z^2$ extends to a globally defined meromorphic form on $\Sigma$, but any fractional linear transformation of $z$ sufficiently close to identity does the same. It is a very strong property that is sufficient to derive that $z$ is a global meromorphic function itself. Indeed,
it means in particular that
\begin{equation}
	\frac{dz}{(z-t)^2}
\end{equation}
is a global defined meromorphic form for any $t$ small enough, hence
\begin{equation}
	\frac{\frac{dz}{z^2}}{	\frac{dz}{(z-t)^2}} = (1-tz^{-1})^2
\end{equation}
is a globally defined meromorphic function for any $t$ small enough. It is a polynomial in $t$, so the coefficients of $t$ and $t^2$ are globally defined meromorphic functions. Hence, $z^{-1}$ (the coefficient of $t$) and, therefore, $z$ are globally defined meromorphic functions. Lemma~\ref{lem:omega02KP} and Proposition~\ref{prop:omega02KP} are proved.
\end{proof}

\section{Computations with Gaussian integrals} \label{sec:Gaussian}

\subsection{A formal Gaussian integral}

\label{sec:FormalGaussianIntegral}

By a Gaussian integral we mean an expression of the form
\begin{equation}\label{eq:Gaussdef}
\int\!\!\dots\!\!\int P(\xi_1,\dots,\xi_n)e^{-Q(\xi_1,\dots,\xi_n)}d\xi_1\dots d\xi_n,
\end{equation}
where $Q$ is a nondegenerate quadratic form and $P$ is a polynomial, or a power series in auxiliary variable~$\sqrt{\hbar}$ with polynomial coefficients.
We treat Gaussian integrals formally, one can think of them using the following computational rules merely as definitions:
\begin{align}\label{eq:Gaussbasic}
	&\tfrac1 {\sqrt{2\pi}}  \int \xi^k e^{-a \frac{\xi^2}{2}} d\xi \coloneqq \begin{cases}
		0, & k = 1, 3, 5, \dots, \\
		(k-1)!! a^{-\frac {k+1}{2}}, & k=0,2,4,\dots;
	\end{cases}
	\\ \label{eq:DoubleGauss}
	&\tfrac\ii {{2\pi}}  \iint \xi^k\zeta^\ell e^{-\xi\zeta} d\xi d\zeta \coloneqq \delta_{k\ell} k!,
\end{align}
here $(-1)!!=1$. In all computations below one can think of a suitable choice of contours that makes this formulas work, and the terms of this shape appear in the formal $\sqrt{\hbar}$ expansion.

The main tool in dealing with a formal Gaussian integral is its invariance under a change of variables. The allowed changes should consist of a linear invertible transformation plus a series in positive powers of~$\sqrt{\hbar}$ with polynomial coefficients. For example, Equation~\eqref{eq:DoubleGauss} can be recovered from~\eqref{eq:Gaussbasic} by a linear transform reducing the quadratic function~$-\xi\zeta$ to a diagonal form and reducing thereby a multiple integral to a repeated one.

\begin{remark} Note that each change of variables in the Gaussian integral make sense only once one performs the corresponding re-parametrization of the integration cycle (or rather the local germ of the integration cycle near the critical point), with the proper choice of orientation. In order to keep the computations below readable, we decided to omit the discussion of the integration cycles in the computations; it is rather straightforward to reconstruct them at each step.
\end{remark}

In order to illustrate how a typical computation with the Gaussian integrals works, let us explain the meaning of the integral~\eqref{eq:etr},
\begin{equation}\label{eq:etr2}
f(z)=\tfrac{\ii}{\sqrt{2\pi\,\hbar}}
\int f^\vee(\chi)\;
y'(\chi)\;
e^{\frac1\hbar\left(x(z)(y(z)-y(\chi))+\int_z^{\chi}x\,dy\right)}\,d\chi.
\end{equation}
Applying the Taylor expansion of the exponent at the point $\chi=z$ we compute
\begin{equation}
\tfrac{1}{\hbar}\Bigl(x(z)(y(z)-y(\chi))+\int_z^{\chi}x\,y'\,dz\Bigr)=\frac12x'(z)y'(z)\tfrac{(\chi-z)^2}{\hbar}+O((\chi-z)^3).
\end{equation}
This expression is a quadratic form in $\xi=\tfrac{\chi-z}{\sqrt{\hbar}}$ plus higher order corrections involving positive powers of~$\sqrt{\hbar}$. Therefore, applying the shift $\chi=z{+}\sqrt{\hbar}\xi$ we can rewrite this integral as
\begin{equation}\label{eq:f-fvee}
f(z)=\tfrac{\ii}{\sqrt{2\pi}}
\int f^\vee(z{+}\sqrt{\hbar}\,\xi)\;y'(z{+}\sqrt{\hbar}\,\xi)\;
e^{\sum_{k\ge2}\hbar^{\frac{k}{2}-1}a_k(z)\frac{\xi^k}{k!}}
d\xi,
\end{equation}
where
\begin{equation}
	a_k(z)\coloneqq\partial_z^{k-1}(xy')-xy^{(k)}.
\end{equation}
The summand in the exponent with $k=2$ is quadratic in~$\xi$ and the higher order terms contain positive powers of~$\sqrt{\hbar}$. Thus, the obtained integral takes a standard form of Gaussian integral and can be computed in the expansion in $\sqrt{\hbar}$ using the rule~\eqref{eq:Gaussbasic} with $a=-x'(z)y'(z)$. In fact, the coefficients of the half-integer powers of~$\hbar$ involve odd powers of $\xi$ only and do not contribute to the integral. As a result, $f(z)$ is a series in nonnegative integer powers of $\hbar$ whose coefficients are certain rational combinations of derivatives of different order of the functions $f^\vee$, $x$, $y$ and a common factor $\sqrt{y'/x'}$:

\begin{equation}\label{eq:example f}
	f(z)=\frac{\sqrt{y'}}{\sqrt{x'}}\Bigl(
	f^\vee+\bigl(
	\tfrac{f^\vee x^{(3)}}{8 (x^{\prime })^2 y^{\prime }}-\tfrac{f^\vee y^{(3)}}{8 x^{\prime } (y^{\prime })^2}-\tfrac{5 f^\vee (x^{\prime\prime })^2}{24 (x^{\prime })^3 y^{\prime }}+\tfrac{f^\vee x^{\prime\prime } y^{\prime\prime }}{24 (x^{\prime })^2 (y^{\prime })^2}+\tfrac{f^\vee (y^{\prime\prime })^2}{6 x^{\prime } (y^{\prime })^3}+\tfrac{f^{\vee\prime } x^{\prime\prime }}{2 (x^{\prime })^2 y^{\prime }}-\tfrac{f^{\vee\prime\prime }}{2 x^{\prime } y^{\prime }}
	\bigr)\hbar+O(\hbar^2)
	\Bigr).
\end{equation}

\begin{remark}\label{Rmk6.1}
From \eqref{eq:f-fvee} we see that the integral transform sending $f^\vee$ to $f$ depends on $x'(z)$ and $y'(z)$ only, or, equivalently, on $dx$ and $dy$.

This transform acts pointwise in a sense that the value of $f$ at the point~$z$ depends on the values of $f^\vee$ in a vicinity of the point~$z$ only. The value $f(z)$ is not always defined at zeros and poles of $dx$ and $dy$, and in the complement to these points the coefficient of each particular power of~$\hbar$ of this transform acts as a finite order differential operator with holomorphic coefficients.

To be precise, Equation~\eqref{eq:Gaussbasic} shows that $f$ involves negative half-integer powers of $x'$ and $y'$. In other words, for any rational exponents~$a$ and~$b$ the function $f^\vee(x')^a(y')^b$ is globally meromorphic if and only if $f(x')^{a+1/2}(y')^{b+1/2}$ is globally meromorphic on~$\Sigma$.
\end{remark}

Let us now apply a change of variables from the local coordinate $\chi$ to the local coordinate $\phi$ defined by an implicit equation
\begin{equation}\label{eq:chi-to-phi}
y(\chi)=y(z)+\sqrt{\hbar} \phi.
\end{equation}
More explicitly, this change is given by $\chi=e^{\phi\sqrt\hbar\partial_y}z$ and its action on any function $g$ is given by $g(\chi)=e^{\phi\sqrt\hbar\partial_y}g(z)$. With this change the integral~\eqref{eq:etr2} takes the following form
\begin{equation}\label{eq:f-fveey}
\begin{aligned}
f(z)&=\tfrac{\ii}{\sqrt{2\pi}}
\int f^\vee(\chi)\;
e^{\sum_{k\ge2}\hbar^{\frac{k}{2}-1}\partial_y^{k-1}x\frac{\phi^k}{k!}}
d\phi,\qquad
f^\vee(\chi)=e^{\phi\sqrt\hbar\partial_y}f^\vee(z).
\end{aligned}
\end{equation}
Here we again use the Taylor expansion of the function $g(\chi) = \int\limits_z^{\chi} xdy$. The term in the exponent with $k=2$ is quadratic in~$\phi$ and the higher order terms involve positive powers of~$\hbar$. Therefore, we obtained a standard form of a Gaussian integral~\eqref{eq:Gaussbasic} with $a=-\frac{dx}{dy}$.

Alternatively, we can use another coordinate~$\psi$ given by
\begin{equation}\label{eq:chi-to-psi}
x(\chi)=x(z)+\sqrt\hbar\psi,
\end{equation}
and we compute for this coordinate applying integration by parts in the exponent
\begin{equation}\label{eq:f-fveex}
\begin{aligned}
f(z)&=\tfrac{\ii}{\sqrt{2\pi}}
\int f^\vee(\chi)\;
y'(\chi)\;
e^{-\frac1\hbar\left((x(z)-x(\chi))y(\chi))+\int_z^{\chi}y\,dx\right)}\,d\chi
\\&=\tfrac{\ii}{\sqrt{2\pi}}
\int f^\vee(\chi)\;
\frac{y'(\chi)}{x'(\chi)}\;
e^{\sum_{k\ge2}\hbar^{\frac{k}{2}-1}(k-1)\partial_x^{k-1}y\frac{\psi^k}{k!}}
d\psi,
\\ f^\vee(\chi)\;
\frac{y'(\chi)}{x'(\chi)}&=
e^{\psi\sqrt\hbar\partial_x}\Bigl(f^\vee(z)\;
\frac{y'(z)}{x'(z)}\Bigr).
\end{aligned}
\end{equation}
Equivalence of~\eqref{eq:f-fvee}, \eqref{eq:f-fveey}, and~\eqref{eq:f-fveex} follows from invariance of Gaussian integrals under the changes. Note that~\eqref{eq:f-fveey} shows, in particular, that the transformation sending $f^\vee$ to $f$ is independent of the choice of the coordinate~$z$.

In a similar way, the Gaussian integral realizing the inverse transformation~\eqref{eq:etri} can be represented in either of the following equivalent forms
\begin{equation}
\begin{aligned}
f^\vee(z)&=\tfrac{1}{\sqrt{2\pi}}
\int f(z{+}\sqrt\hbar\xi)\;
x'(z{+}\sqrt\hbar\xi)\;
e^{-\sum_{k\ge2}\hbar^{\frac{k}{2}-1}a^\vee_k(z)\frac{\xi^k}{k!}}\,d\xi
\\&=\tfrac{1}{\sqrt{2\pi}}
\int f(\chi)\;
e^{-\sum_{k\ge2}\hbar^{\frac{k}{2}-1}\partial_x^{k-1}y\frac{\psi^k}{k!}}
d\psi
\\&=
\tfrac{1}{\sqrt{2\pi}}
\int f(\chi)\;
\frac{x'(\chi)}{y'(\chi)}\;
e^{-\sum_{k\ge2}\hbar^{\frac{k}{2}-1}(k-1)\partial_y^{k-1}x\frac{\phi^k}{k!}}
d\phi,
\end{aligned}
\end{equation}
where $\chi$, $\psi$, and $\phi$, are related by the changes~\eqref{eq:chi-to-phi} and~\eqref{eq:chi-to-psi}, and
\begin{equation}
		a^\vee_k(z)\coloneqq\partial_z^{k-1}(yx')-yx^{(k)}.
\end{equation}

All other Gaussian integrals considered in the previous sections are treated in a similar way. Note, for example, that the transformation of Theorem~\ref{thm:transformation-KKernel} represented by double Gaussian integral involves factors of the form $\sqrt{x'_1x'_2y'_1y'_2}$ that have a well defined branch if $z_1$ and $z_2$ are close to one another so that this transformation is defined globally in a vicinity of the diagonal in~$\Sigma^2$.

\begin{proof}[Proof of Proposition~\ref{prop:identity-Gauss1}]
The composition of transformations~\eqref{eq:etr} and~\eqref{eq:etri} is represented, by construction, by the double Gaussian integral
\begin{equation}
\begin{aligned}
\tfrac{\ii}{2\pi\,\hbar}&
\iint
f^\vee(\zeta)\;
y'(\zeta)\;
x'(\chi)\;
e^{\frac1\hbar\left(x(\chi)(y(\chi)-y(\zeta))+\int_\chi^{\zeta}x\,dy\right)}
e^{-\frac1\hbar\left(y(z)(x(z)-x(\chi))+\int_z^{\chi}y\,dx\right)}\,d\zeta\,d\chi
\\&=\tfrac{\ii}{2\pi\,\hbar}
\iint
f^\vee(\zeta)\;
y'(\zeta)\;
x'(\chi)\;
e^{\frac1\hbar\left(-x(\chi)(y(\zeta)-y(z))+\int_z^{\zeta}x\,dy\right)}\,d\zeta\,d\chi
\\&=\tfrac{\ii}{2\pi\,\hbar}
\iint
f^\vee(\zeta)\;
e^{-\phi\psi+\frac1\hbar\left(-x(z)(y(\zeta)-y(z))+\int_z^{\zeta}x\,dy\right)}
d\phi\,d\psi,
\end{aligned}
\end{equation}
where in the second line we use integration by parts and in the third line changes of local coordinates defined by the implicit equations $y(\zeta)=y(z)+\sqrt\hbar \phi$, and $x(\chi)=x(z)+\sqrt\hbar\psi$.
The quadratic part of the exponent is $-\phi\bigl(\psi-\frac12\frac{dx}{dy}\phi\bigr)$. One can make a further change considering $\eta=\psi-\frac12\frac{dx}{dy}\phi$ as a new variable instead of $\psi$. Thus, we reduced the integral to the form~\eqref{eq:DoubleGauss}. The factor in front of $e^{-\phi\eta}d\phi d\eta$ does not involve $\eta$, and by~\eqref{eq:DoubleGauss}, the integral is equal to the value of this factor at $\phi=0$ which is $f^\vee(z)$.
\end{proof}

\subsection{Key identity} The main tool relating different forms of $x-y$ duality transformation is contained in the following identity.

\begin{lemma}\label{lem:key}
Let $y=y(z)$ be a meromorphic function and $g(z,u)$ be a series in $\hbar$ whose coefficients are meromorphic functions on the spectral curve depending on an additional parameter $u$ in a polynomial way. Assume that $g(z,u/\sqrt\hbar)$ involves non-negative powers of $\hbar$ only. Then the following identity holds
\begin{equation}\label{eq:key}
\sum_{r\ge0}\partial_y^r[u^r]g(z,u)=\frac{\ii}{2\pi}\iint g(\zeta,u)\,y'(\zeta)\,e^{-u\,(y(\zeta)-y(z))}du\,d\zeta.
\end{equation}
The double integral on the right hand side is understood in the sense of asymptotic expansion at small absolute value of $\hbar$ at $\zeta=z$.
\end{lemma}

\begin{proof}
Using~\eqref{eq:DoubleGauss} we compute:
\begin{equation}
\begin{aligned}
\sum_{r\ge0}\partial_y^r&[u^r]g(z,u)
=\sum_{r\ge0}\hbar^{r/2}\partial_y^r[u^r]g(z,u/\sqrt\hbar)
\\&=\sum_{r\ge0}\hbar^{r/2}\partial_y^r\frac{\ii}{2\pi}\iint\frac{v^r}{r!}g(z,u/\sqrt\hbar)\,e^{-u\,v}du\,dv
\\&=\frac{\ii}{2\pi}\iint e^{v\sqrt\hbar\partial_y}g(z,u/\sqrt\hbar)\,e^{-u\,v}du\,dv.
\\&=\frac{\ii}{2\pi}\iint g(\zeta,u/\sqrt\hbar)\,e^{-u\,v}du\,dv
\\&=\frac{\ii}{2\pi}\iint g(\zeta,u)\,e^{-u\,\sqrt\hbar v}du\,\sqrt\hbar dv.
\end{aligned}
\end{equation}
where $\zeta=\zeta(z,v)=e^{v\sqrt\hbar\partial_y}z$. Let us consider an expression for $\zeta$ in terms of $v$ as a change of variables in the integral. Equivalently, this change is determined by an implicit equation $y(\zeta)=y(z)+v\sqrt\hbar$ implying that $y'(\zeta)d\zeta=\sqrt\hbar\,dv$. The obtained integral expression written in coordinates~$u$ and~$\zeta$ takes the form exactly as in Lemma.
\end{proof}

\begin{remark}
A rescaling of the parameter~$u$ involving non-integer powers of~$\hbar$ is needed for the change of variables in the Gaussian integral to make sense. Recall that an allowed change involves a linear invertible transformation plus a series in positive powers of~$\hbar$ with polynomial coefficients. In that sense, the actual variables of integration in the integral~\eqref{eq:key} for which this integral takes a standard form~\eqref{eq:Gaussdef} of a formal Gaussian integral are $\bigl(\sqrt\hbar\,u,\frac{\zeta-z}{\sqrt\hbar}\bigr)$
\end{remark}

\subsection{Identities with Gaussian integrals}
In this section we formulate and prove an important set of identities needed for the proof of Theorem~\ref{thm:transformation-Omega}.
Let us have a look at the right square of Diagram~\eqref{eq:XYdiagram}. Observe that all its transformations are products of independent transformations acting on variables indexed by $1,2,\dots,n$. Therefore, in order to prove commutativity of this square it is sufficient to prove commutativity of the corresponding transformations for one particular index. The commutativity of this reduced set of transformations acting on functions in just two variables is checked in this section.

Consider the following diagram of transformations relating four functions in two variables
\begin{align}\label{eq:f-F-fvee-Fvee-Diagram}
	\xymatrix@C=30pt{
		f(z,u) \ar @<0.5ex> [d]  & F(w,\bar w) \ar [l]  \ar @<0.5ex> [d] \\
		f^\vee(z,v)\ar @<0.5ex> [u] & F^\vee(w,\bar w) \ar[l]  \ar @<0.5ex> [u]
	}
\end{align}
where the transformations represented by arrows are defined by
	\begin{align}\label{eq:F-to-f}
		f (z,u) & =e^{- u(\cS(u\hbar\partial_{x})-1) y}F\Big(e^{\frac{\hbar u}2 \partial_{x}}z,e^{-\frac{\hbar u}2 \partial_{x}}z\Big),
\\\label{eq:Fv-to-fv}
				f^\vee(z,v) & = e^{- v(\cS(v\hbar\partial_{y})-1) x}   F^\vee\Big(e^{\frac{\hbar v}2 \partial_{y}}z,e^{-\frac{\hbar v}2 \partial_{y}}z\Big),
\\\label{eq:fv-to-f}
	f(z, u) & = -e^{u y }\sum_{r=0}^\infty \partial_{x}^{r} \; e^{ -u y} \frac{dy}{dx}  [v^{r}] f^{\vee} (z,v),
	\\\label{eq:f-to-fv}
	f^\vee (z, v) & =  -e^{v x }\sum_{r=0}^\infty \partial_{y}^{r} \; e^{ -v x} \frac{dx}{dy} [u^{r}] f (z, u),
\\\label{eq:Fv-to-F}
	F(w,\bar w)
&= \frac{-\ii}{2\pi\hbar}
	\iint
	F^\vee(\chi,\bar \chi)
	y'(\chi)y'(\bar \chi)    d\chi d\bar \chi\times
	\\ \notag &\qquad
	e^{-\frac{1}{\hbar}\big(x(\bar w)(y(\bar w)-y(\chi))+\int_{\bar w}^{\chi} xdy\big)}
                  e^{\frac{1}{\hbar}\big(x(w)(y(w)-y(\bar \chi))+\int_{w}^{\bar \chi} xdy\big)},
	\\\label{eq:F-to-Fv}
	F^\vee(w,\bar w)
&=\frac{-\ii}{2\pi\hbar}
\iint
F(\chi,\bar \chi)
{x'(\chi)x'(\bar \chi)}    d\chi d\bar \chi\times
\\ \notag &\qquad
 e^{-\frac{1}{\hbar}\big(y(\bar w)(x(\bar w)-x(\chi))+\int_{\bar w}^{\chi} ydx\big)}
              e^{\frac{1}{\hbar}\big(y(w)(x(w)-x(\bar \chi))+\int_{w}^{\bar \chi} ydx\big)}.
\end{align}

\begin{proposition}\label{prop:DiagfF}
The transformations forming Diagram~\eqref{eq:f-F-fvee-Fvee-Diagram} commute.
\end{proposition}

\begin{proof}
The fact that the transformations relating~$F$ and $F^\vee$ are inverse to one another follows from Proposition~\ref{prop:identity-Gauss1} proved in the previous section. The fact that the transformations \eqref{eq:fv-to-f} and \eqref{eq:f-to-fv} relating $f$ and $f^\vee$ are inverse to one another is proved implicitly in~\cite{alexandrov2022universal}. It follows also from the commutativity of the remaining part of the diagram since the transformations represented by horizontal arrows are invertible (at least in the asymptotic expansion in~$\hbar$). Thus, the most essential part of the proof of Proposition is the statement that the vertical transformations of the diagram commute with the horizontal ones. To prove that, we assume that $F$ is an arbitrary function (that is, an arbitrary series in~$\hbar$ whose coefficients are meromorphic functions on~$\Sigma^2$ defined in a vicinity of the diagonal). Define~$f$ and~$F^\vee$ by~\eqref{eq:F-to-f} and~\eqref{eq:F-to-Fv}, respectively, and we have to prove that the two expressions for $f^\vee$ given by~\eqref{eq:f-to-fv} and~\eqref{eq:Fv-to-fv} agree.

By Lemma~\ref{lem:key} an expression~\eqref{eq:f-to-fv} for $f^\vee$ can be represented as the following Gaussian integral
\begin{equation}\label{eq:fv-Gauss1}
\frac{-\ii}{2\pi}\iint
e^{ v\,(x(z)- x(\zeta))} f(\zeta,u)\,x'(\zeta)\,e^{-u\,(y(\zeta)-y(z))}du\,d\zeta.
\end{equation}

On the other hand, an expression for~$f^\vee$ implied by~\eqref{eq:F-to-Fv} and~\eqref{eq:Fv-to-fv} can be represented by the following Gaussian integral
\begin{equation}\label{eq:fv-Gauss2}
\begin{aligned}
e^{- v(\cS(v\hbar\partial_{y})-1) x}&\restr{(w,\bar w)}{(e^{\frac{\hbar v}2 \partial_{y}}z,e^{-\frac{\hbar v}2 \partial_{y}}z)}
\frac{-\ii}{2\pi\hbar}
\iint
F(\chi,\bar \chi)
{x'(\chi)x'(\bar \chi)}    d\chi d\bar \chi\times
\\&\qquad
 e^{-\frac{1}{\hbar}\big(y(\bar w)(x(\bar w)-x(\chi))+\int_{\bar w}^{\chi} ydx\big)}
              e^{\frac{1}{\hbar}\big(y(w)(x(w)-x(\bar \chi))+\int_{w}^{\bar \chi} ydx\big)}
 \\=&
 \restr{(w,\bar w)}{(e^{\frac{\hbar v}2 \partial_{y}}z,e^{-\frac{\hbar v}2 \partial_{y}}z)}
\frac{-\ii}{2\pi\hbar}
\iint
F(\chi,\bar \chi)
{x'(\chi)x'(\bar \chi)}    d\chi d\bar \chi\times
\\ &\qquad
e^{\frac{1}{\hbar}\big(\hbar v x(z)-\int_{\bar w}^w xdy
 -y(\bar w)(x(\bar w)-x(\chi))-\int_{\bar w}^{\chi} ydx
              +y(w)(x(w)-x(\bar \chi))+\int_{w}^{\bar \chi} ydx\big)}
 \\=&
 \restr{(w,\bar w)}{(e^{\frac{\hbar v}2 \partial_{y}}z,e^{-\frac{\hbar v}2 \partial_{y}}z)}
\frac{-\ii}{2\pi\hbar}
\iint
F(\chi,\bar \chi)
{x'(\chi)x'(\bar \chi)}    d\chi d\bar \chi\times
\\ &\qquad
e^{\frac{1}{\hbar}\big(\hbar v x(z)
 +y(\bar w)x(\chi)
              -y(w)x(\bar \chi)-\int_{\bar\chi}^{\chi} ydx\big)}
\\=&
\frac{-\ii}{2\pi\hbar}
\iint
F(\chi,\bar \chi)
{x'(\chi)x'(\bar \chi)}    d\chi d\bar \chi\times
\\ &\qquad
e^{v (x(z)
 -\frac{1}{2}x(\chi)
              -\frac{1}{2}x(\bar\chi))+\frac{1}{\hbar}\big(y(z)(x(\chi)-x(\bar\chi))-\int_{\bar\chi}^{\chi} ydx\big)}.
\end{aligned}
\end{equation}
Here in the first equality we use the identity $v \hbar\cS(v \hbar \partial_y) x=\restr{(w,\bar w)}{(e^{\frac{\hbar v}2 \partial_{y}}z,e^{-\frac{\hbar v}2 \partial_{y}}z)}\int\limits_{\bar w} ^{w} x dy$. In the second equality we again use an integration by parts and in the last equality we used that for the specified substitution we have $y(w)=y(z)+\frac{v\hbar}{2}$, $y(\bar w)=y(z)-\frac{v\hbar}{2}$.

In order to compare the obtained integral with~\eqref{eq:fv-Gauss1} we apply the change of variables defined by $(\chi,\bar\chi)=(e^{\frac{u\hbar}{2}\frac{d}{dx(\zeta)}}\zeta, e^{-\frac{u\hbar}{2}\frac{d}{dx(\zeta)}}\zeta)$. For this change we find
\begin{align}
x(\chi)=x(\zeta)+\tfrac{u\hbar}{2},\quad
x(\bar\chi)=x(\zeta)-\tfrac{u\hbar}{2},
\\x'(\chi)x'(\bar\chi)\,d\chi\,d\bar\chi=
\hbar\,x'(\zeta)d\zeta\,du.
\end{align}
Then taking $F(\chi,\bar\chi)$ from \eqref{eq:F-to-f} and using the identity
$
u \hbar\cS(u \hbar \partial_x) y=\restr{(\chi,\bar \chi)}{(e^{\frac{\hbar u}2 \partial_{x}}z,e^{-\frac{\hbar u}2 \partial_{x}}z)}\int\limits_{\bar \chi} ^{\chi} y dx	
$
 we obtain
\begin{equation}
	F(\chi,\bar\chi)=e^{-u\,y(\zeta)+\frac1{\hbar}\int_{\bar\chi}^\chi y\,dx}f(\zeta,u).
\end{equation}
And finally the integral on the right hand side of~\eqref{eq:fv-Gauss2} takes exactly the form~\eqref{eq:fv-Gauss1}. This completes the proof of Proposition~\ref{prop:DiagfF}.
\end{proof}

\begin{remark}
To make sense to the computations in the proof, we imply that all involved variables are regarded with certain scales. In order to da that the integrals involved take the standard form~\eqref{eq:Gaussdef} of formal Gaussian integral on each step of computations. For example, for the integrals~\eqref{eq:F-to-f} and~\eqref{eq:Fv-to-fv} we assume that $\frac{\chi-\bar w}{\sqrt\hbar}=O(\hbar^0)$, $\frac{\bar\chi-w}{\sqrt\hbar}=O(\hbar^0)$. For the integral~\eqref{eq:fv-Gauss1} we assume $u\sqrt\hbar=O(\hbar^{0})$, $\frac{\zeta-z}{\sqrt\hbar}=O(\hbar^0)$. With these scales all the changes have allowed form of a linear transform plus terms of positive order in~$\hbar$ with polynomial coefficients.
\end{remark}

\section{Refined \texorpdfstring{$x-y$}{x--y} duality} \label{sec:ReformulationXY-duality}

In this section we give the proofs of all relations in Theorem~\ref{thm:transformation-Omega}, using the results of~\cite{alexandrov2022universal} and Proposition~\ref{prop:DiagfF} above.

\begin{proof}[Proof of Theorem~\ref{thm:transformation-Omega}]

First, we show how the first two relations of  Theorem~\ref{thm:transformation-Omega} follow directly from~\cite[Proposition 4.8]{alexandrov2022universal}.	
	
 Recall the definition of $\bW^{(g)}_n$ and notation $\cW^{x,(g)}_{m+1,0}$ in~\cite{alexandrov2022universal}. For each $i=1,\dots,n$ we can redevelop the sum over graphs in $\bW^{(g)}_n$ to become a homogeneous linear combination of $e^{\tilde u_ix_i}\cW^{x,(\tilde g)}_{\tilde m+1,0} (\tilde u_i; \tilde z_{\llbracket \tilde m \rrbracket}; z_i; \emptyset)$. The variables $\tilde z_{\llbracket m \rrbracket}$ in this case are further specialized via attachment of the corresponding multi-edges to the vertices labeled by $j\not= i$. \cite[Equation (4.7)]{alexandrov2022universal} implies that the $i$-th factor in Equation~\eqref{eq:bW-vee-from-bW} transforms each of these $e^{\tilde u_ix_i}\cW^{x,(\tilde g)}_{\tilde m+1,0} (\tilde u_i; \tilde z_{\llbracket m \rrbracket}; z_i;\emptyset)$ into $e^{u_iy_i}\cW^{y,(\tilde g)}_{\tilde m,1} (u_i; \tilde z_{\llbracket m \rrbracket}; z_i;\emptyset)$.
	
Applying this observation for each $i=1,\dots,n$, we obtain a sum over the same graphs that we can redevelop into a homogeneous linear combination of $e^{ u_iy_i}\cW^{y,(\tilde g)}_{0, 1+\tilde m} (u_i; \emptyset; z_i;\tilde z_{\llbracket \tilde m \rrbracket})$ for each $i=1,\dots,n$, which then coincides with the same linear expansion of $\bW^{\vee,(g)}_n$. This proves Equation~\eqref{eq:bW-vee-from-bW}. Equation~\eqref{eq:bW-from-bW-vee} is fully analogous.

Now the relation between $\Omega_n$ and $\bW_n = \sum_{g=0}^\infty \hbar^{2g-2+n}\bW^{(g)}_n $ is almost tautological in the framework of Proposition~\ref{prop:DiagfF}. Indeed, $\bW_n (z_{\llbracket n \rrbracket}, u_{\llbracket n \rrbracket})$  is related to $\Omega_n(w_{\llbracket n \rrbracket},\bar w_{\llbracket n \rrbracket})$
by Equation~\eqref{eq:Omega-to-bW}, which is exactly the application of Equation~\eqref{eq:F-to-f} in each pair of the variables $w_i,\bar w_i$. The same holds for $\Omega^\vee_n$ and $\bW^\vee_n$, cf.~Equations~\eqref{eq:Omega-to-bW-vee} and~\eqref{eq:Fv-to-fv}. Thus Proposition~\ref{prop:DiagfF} applied in each pair of variables $w_i,\bar w_i$, $i=1,\dots,n$ implies that Equations~\eqref{eq:Omega-int} and~\eqref{eq:Omega-vee-int} use as vertical arrows between $\Omega_n$ and $\Omega^\vee_n$ in
\begin{equation}\label{eq:XYdiagram-W-Omega}
	\begin{aligned}
		{\ }   \\
		\xymatrix@C=30pt{
		 \{\bW^{(g)}_n\}   \ar  @<0.5ex> [d] & \{\Omega_n\} \ar[l] \ar @<0.5ex> [d] \\
		  \{\bW^{\vee,(g)}_n\}    \ar @<0.5ex> [u]& \{\Omega^{\vee}_n\} \ar[l] \ar @<0.5ex> [u]
		}
		\\
		{\ }
	\end{aligned}
\end{equation}
complete this diagram to a commutative one. Then~\eqref{eq:XY-diag-triv} implies that the horizontal arrows in~\eqref{eq:XYdiagram-W-Omega} do not lose any information, that is, the commutativity of~\eqref{eq:XYdiagram-W-Omega} implies that Equations~\eqref{eq:Omega-int} and~\eqref{eq:Omega-vee-int} hold.
\end{proof}

\printbibliography

@article {ABDKS4,
	AUTHOR = {Alexandrov, A. and Bychkov, B. and Dunin-Barkowski, P. and
	Kazarian, M. and Shadrin, S.},
	title = {Log topological recursion through the prism of x-y swap},
	JOURNAL = {Int. Math. Res. Not. IMRN},
	FJOURNAL = {International Mathematics Research Notices. IMRN},
	YEAR = {2024},
	NUMBER = {21},
	PAGES = {13461--13487},
	ISSN = {1073-7928,1687-0247},
	MRCLASS = {14H81},
	MRNUMBER = {4819863},
	DOI = {10.1093/imrn/rnae213},
	URL = {https://doi.org/10.1093/imrn/rnae213},
	
	archiveprefix = {arXiv},
	eprint = {2312.16950},
	primaryclass = {math-ph},
}

@misc{alexandrov2022universal,
	archiveprefix = {arXiv},
	author = {Alexander Alexandrov and Boris Bychkov and Petr Dunin-Barkowski and Maxim Kazarian and Sergey Shadrin},
	eprint = {2212.00320},
	primaryclass = {math-ph},
	title = {A universal formula for the $x-y$ swap in topological recursion},
	year = {2025},
	note   =  {To appear in J. Eur. Math. Soc.},
}

@article {alexandrov2023higher,
	AUTHOR = {Alexandrov, Alexander and Dhara, Saswati},
	TITLE = {On {H}igher {B}r\'ezin--{G}ross--{W}itten {T}au {F}unctions},
	JOURNAL = {Int. Math. Res. Not. IMRN},
	FJOURNAL = {International Mathematics Research Notices. IMRN},
	YEAR = {2025},
	NUMBER = {2},
	PAGES = {rnae286},
	ISSN = {1073-7928,1687-0247},
	MRCLASS = {99-06},
	MRNUMBER = {4852745},
	DOI = {10.1093/imrn/rnae286},
	URL = {https://doi.org/10.1093/imrn/rnae286},
	
	archiveprefix = {arXiv},
	eprint = {2204.12273},
	primaryclass = {math-ph},
}

@misc{limits,
	title={Taking limits in topological recursion}, 
	author={Gaëtan Borot and Vincent Bouchard and Nitin Kumar Chidambaram and Reinier Kramer and Sergey Shadrin},
	year={2023},
	eprint={2309.01654},
	archivePrefix={arXiv},
	primaryClass={math.AG}
}

@article{borot-comb,
	eprint = {1904.02267},
	primaryclass = {math.CO},
		archiveprefix = {arXiv},
	author = {Borot, Ga{\"e}tan and Charbonnier, S{\'e}verin and Do, Norman and Garcia-Failde, Elba},
	fjournal = {The Electronic Journal of Combinatorics},
	issn = {1077-8926},
	journal = {Electron. J. Comb.},
	language = {English},
	number = {3},
	pages = {research paper p3.43, 24},
	title = {Relating ordinary and fully simple maps via monotone {Hurwitz} numbers},
	url = {www.combinatorics.org/ojs/index.php/eljc/article/view/v26i3p43},
	volume = {26},
	year = {2019},
	zbl = {1430.05004},
	doi = {10.37236/8634},
	zbmath = {7104367}}

@article {borot2021topological,
	archiveprefix = {arXiv},
eprint = {2106.09002},
primaryclass = {math.CO},

AUTHOR = {Borot, Ga\"etan and Charbonnier, S\'everin and Garcia-Failde,
Elba},
TITLE = {Topological recursion for fully simple maps from ciliated
maps},
JOURNAL = {Comb. Theory},
FJOURNAL = {Combinatorial Theory},
VOLUME = {4},
YEAR = {2024},
NUMBER = {2},
PAGES = {Paper No. 14, 33},
ISSN = {2766-1334},
MRCLASS = {05C10 (05C30 46L54)},
MRNUMBER = {4807153},
}

@misc{borot2021analytic,
	title={Functional relations for higher-order free cumulants}, 
	author={Gaëtan Borot and Séverin Charbonnier and Elba Garcia-Failde and Felix Leid and Sergey Shadrin},
	year={2021},
	eprint={2112.12184},
	archivePrefix={arXiv},
	primaryClass={math.OA},
}

@article{BGF-conjecture,
	author = {Borot, Ga{\"e}tan and Garcia-Failde, Elba},
	doi = {10.1007/s00220-020-03867-1},
	fjournal = {Communications in Mathematical Physics},
	issn = {0010-3616},
	journal = {Comm. Math. Phys.},
	language = {English},
	number = {2},
	pages = {581--654},
	title = {Simple maps, {Hurwitz} numbers, and topological recursion},
	volume = {380},
	year = {2020},
	zbl = {1457.37055},
	zbmath = {7286852},
	eprint={1710.07851},
	archivePrefix={arXiv},
	primaryClass={math-ph}
}

@article{borot2022higher,
	eprint={2010.03512},
	archivePrefix={arXiv},
	primaryClass={math-ph},

	AUTHOR = {Borot, Ga\"etan and Kramer, Reinier and Sch\"uler, Yannik},
	TITLE = {Higher {A}iry structures and topological recursion for
	singular spectral curves},
	JOURNAL = {Ann. Inst. Henri Poincar\'e{} D},
	FJOURNAL = {Annales de l'Institut Henri Poincar\'e{} D. Combinatorics,
	Physics and their Interactions},
	VOLUME = {11},
	YEAR = {2024},
	NUMBER = {1},
	PAGES = {1--146},
	ISSN = {2308-5827,2308-5835},
	MRCLASS = {81R10 (14H81 14N10 17B69 81T45)},
	MRNUMBER = {4705584},
	DOI = {10.4171/aihpd/168},
	URL = {https://doi.org/10.4171/aihpd/168},
}

@article {B-E,
	AUTHOR = {Bouchard, Vincent and Eynard, Bertrand},
	TITLE = {Think globally, compute locally},
	JOURNAL = {J. High Energy Phys.},
	FJOURNAL = {Journal of High Energy Physics},
	YEAR = {2013},
	NUMBER = {2},
	PAGES = {143, front matter + 34},
	ISSN = {1126-6708,1029-8479},
	MRCLASS = {81T45 (30F10 32G15)},
	MRNUMBER = {3046532},
	MRREVIEWER = {Lee-Peng\ Teo},
	DOI = {10.1007/JHEP02(2013)143},
	URL = {https://doi.org/10.1007/JHEP02(2013)143},
	eprint = {1211.2302},
	archiveprefix = {arXiv},
	primaryClass={math-ph},
}

@article {BDKS-first,
	AUTHOR = {Bychkov, Boris and Dunin-Barkowski, Petr and Kazarian, Maxim
	and Shadrin, Sergey},
	TITLE = {Explicit closed algebraic formulas for {O}rlov-{S}cherbin
	{$n$}-point functions},
	JOURNAL = {J. \'{E}c. polytech. Math.},
	FJOURNAL = {Journal de l'\'{E}cole polytechnique. Math\'{e}matiques},
	VOLUME = {9},
	YEAR = {2022},
	PAGES = {1121--1158},
	ISSN = {2429-7100},
	MRCLASS = {05E14 (14H30 14N10 33C80 37K10)},
	MRNUMBER = {4453412},
	MRREVIEWER = {Piotr G. Grinevich},
	DOI = {10.5802/jep.202},
	URL = {https://doi.org/10.5802/jep.202},
	eprint={2008.13123},
	archivePrefix={arXiv},
	primaryClass={math.CO}
}

@article{FullySimpleProof,
	eprint = {2106.08368},
	archivePrefix={arXiv},
	primaryClass={math-ph},
	author = {Bychkov, Boris and Dunin-Barkowski, Petr and Kazarian, Maxim and Shadrin, Sergey},
	doi = {10.1007/s00220-023-04732-7},
	fjournal = {Communications in Mathematical Physics},
	issn = {0010-3616},
	journal = {Comm. Math. Phys.},
	language = {English},
	number = {1},
	pages = {665--694},
	title = {Generalised ordinary vs fully simple duality for {{\(n\)}}-point functions and a proof of the {Borot}-{Garcia}-{Failde} conjecture},
	volume = {402},
	year = {2023},
	zbmath = {7719665}}

@article {ChekhovNorbury,
	AUTHOR = {Chekhov, Leonid and Norbury, Paul},
	TITLE = {Topological recursion with hard edges},
	JOURNAL = {Internat. J. Math.},
	FJOURNAL = {International Journal of Mathematics},
	VOLUME = {30},
	YEAR = {2019},
	NUMBER = {3},
	PAGES = {1950014, 29},
	ISSN = {0129-167X,1793-6519},
	MRCLASS = {14N10 (14H81 32G15)},
	MRNUMBER = {3941980},
	MRREVIEWER = {Piotr\ Su\l kowski},
	DOI = {10.1142/S0129167X19500149},
	URL = {https://doi.org/10.1142/S0129167X19500149},
	eprint={1702.08631},
	archivePrefix={arXiv},
	primaryClass={math.AG}
}

@misc{chidambaram2023relations,
	archiveprefix = {arXiv},
	author = {Nitin Kumar Chidambaram and Elba Garcia-Failde and Alessandro Giacchetto},
	eprint = {2205.15621},
	primaryclass = {math.AG},
	title = {Relations on $\overline{\mathcal{M}}_{g,n}$ and the negative $r$-spin Witten conjecture},
	year = {2023}}

@article {charbonnier2022shifted,
	eprint={2203.16523},
	archivePrefix={arXiv},
	primaryClass={math.AG},
	
	AUTHOR = {Charbonnier, S\'everin and Chidambaram, Nitin Kumar and
	Garcia-Failde, Elba and Giacchetto, Alessandro},
	TITLE = {Shifted {W}itten classes and topological recursion},
	JOURNAL = {Trans. Amer. Math. Soc.},
	FJOURNAL = {Transactions of the American Mathematical Society},
	VOLUME = {377},
	YEAR = {2024},
	NUMBER = {2},
	PAGES = {1069--1110},
	ISSN = {0002-9947,1088-6850},
	MRCLASS = {14H10 (14H70 34E05 81R10)},
	MRNUMBER = {4688543},
	DOI = {10.1090/tran/9046},
	URL = {https://doi.org/10.1090/tran/9046},
}

@article {EO,
	AUTHOR = {Eynard, Bertrand and Orantin, Nicolas},
	TITLE = {Topological recursion in enumerative geometry and random
	matrices},
	JOURNAL = {J. Phys. A},
	FJOURNAL = {Journal of Physics. A. Mathematical and Theoretical},
	VOLUME = {42},
	YEAR = {2009},
	NUMBER = {29},
	PAGES = {293001, 117},
	ISSN = {1751-8113,1751-8121},
	MRCLASS = {14N10 (14H10 14N35 37K20)},
	MRNUMBER = {2519749},
	MRREVIEWER = {Hsian-Hua\ Tseng},
	DOI = {10.1088/1751-8113/42/29/293001},
	URL = {https://doi.org/10.1088/1751-8113/42/29/293001},
}

@misc{guo2023generalization,
	title={A generalization of the Witten conjecture through spectral curve}, 
	author={Shuai Guo and Ce Ji and Qingsheng Zhang},
	year={2023},
	eprint={2309.12271},
	archivePrefix={arXiv},
	primaryClass={math-ph}
}

@article{hock2022xy,
	archiveprefix = {arXiv},
	eprint = {2201.05357},
	primaryclass = {math-ph},
	title = {On the $x$-$y$ Symmetry of Correlators in Topological Recursion via Loop Insertion Operator},
	AUTHOR = {Hock, Alexander},
JOURNAL = {Comm. Math. Phys.},
FJOURNAL = {Communications in Mathematical Physics},
VOLUME = {405},
YEAR = {2024},
NUMBER = {7},
PAGES = {Paper No. 166},
ISSN = {0010-3616,1432-0916},
MRCLASS = {46L54 (05 15B57)},
MRNUMBER = {4768536},
DOI = {10.1007/s00220-024-05043-1},
URL = {https://doi.org/10.1007/s00220-024-05043-1},
}

@article {hock2022simple,
AUTHOR = {Hock, Alexander},
TITLE = {A simple formula for the {$x$}-{$y$} symplectic transformation
in topological recursion},
JOURNAL = {J. Geom. Phys.},
FJOURNAL = {Journal of Geometry and Physics},
VOLUME = {194},
YEAR = {2023},
PAGES = {Paper No. 105027, 26},
ISSN = {0393-0440,1879-1662},
MRCLASS = {46L54 (15B52 16R60)},
MRNUMBER = {4659813},
DOI = {10.1016/j.geomphys.2023.105027},
URL = {https://doi.org/10.1016/j.geomphys.2023.105027},

	archiveprefix = {arXiv},
	eprint = {2211.08917},
	primaryclass = {math-ph},
}

@article {hock2023laplace,
	AUTHOR = {Hock, Alexander},
	TITLE = {Laplace transform of the {$x - y$} symplectic transformation
	formula in topological recursion},
	JOURNAL = {Commun. Number Theory Phys.},
	FJOURNAL = {Communications in Number Theory and Physics},
	VOLUME = {17},
	YEAR = {2023},
	NUMBER = {4},
	PAGES = {821--845},
	ISSN = {1931-4523,1931-4531},
	MRCLASS = {14N10 (05A15 14H70 14H81 30F30)},
	MRNUMBER = {4704940},

	archiveprefix = {arXiv},
	eprint = {2304.03032},
	primaryclass = {math-ph},
}

@article{kazarian2021polynomial,
	eprint={2112.11672},
	archivePrefix={arXiv},
	primaryClass={math.AG},
   AUTHOR = {Kazarian, Maxim and Norbury, Paul},
TITLE = {Polynomial relations among kappa classes on the moduli space
of curves},
JOURNAL = {Int. Math. Res. Not. IMRN},
FJOURNAL = {International Mathematics Research Notices. IMRN},
YEAR = {2024},
NUMBER = {3},
PAGES = {1825--1867},
ISSN = {1073-7928,1687-0247},
MRCLASS = {14H10},
MRNUMBER = {4702265},
DOI = {10.1093/imrn/rnad061},
URL = {https://doi.org/10.1093/imrn/rnad061},
}

@article {Kharchev-Marshakov,
	AUTHOR = {Kharchev, S. and Marshakov, A.},
	TITLE = {On {$p$}-{$q$} duality and explicit solutions in {$c\leq 1$}
	{$2$}D gravity models},
	JOURNAL = {Internat. J. Modern Phys. A},
	FJOURNAL = {International Journal of Modern Physics A. Particles and
	Fields. Gravitation. Cosmology},
	VOLUME = {10},
	YEAR = {1995},
	NUMBER = {8},
	PAGES = {1219--1236},
	ISSN = {0217-751X,1793-656X},
	MRCLASS = {81T30 (81T40)},
	MRNUMBER = {1321929},
	MRREVIEWER = {Igor\ Polyubin},
	DOI = {10.1142/S0217751X95000577},
	URL = {https://doi.org/10.1142/S0217751X95000577},
	eprint={hep-th/9303100},
	archivePrefix={arXiv},
}

@book{MJD,
	author = {Miwa, T. and Jimbo, M. and Date, E.},
	edition = {Reprint of the 2000 hardback edition},
	fseries = {Cambridge Tracts in Mathematics},
	isbn = {978-1-107-40419-9},
	issn = {0950-6284},
	language = {English},
	publisher = {Cambridge: Cambridge University Press},
	series = {Camb. Tracts Math.},
	title = {Solitons: differential equations, symmetries and infinite dimensional algebras. {Transl}. from the {Japanese} by {Miles} {Reid}},
	volume = {135},
	year = {2012},
	zbl = {1239.37011},
	zbmath = {6037369}}

@article{MiMoSeff96,
	eprint= {hep-th/9404005},
	archiveprefix={arXiv},
	author = {Mironov, A. and Morozov, A. and Semenoff, G. W.},
	doi = {10.1142/S0217751X96002339},
	fjournal = {International Journal of Modern Physics A},
	issn = {0217-751X},
	journal = {Int. J. Mod. Phys. A},
	language = {English},
	number = {28},
	pages = {5031--5080},
	title = {Unitary matrix integrals in the framework of the generalized {Kontsevich} model.},
	volume = {11},
	year = {1996},
	zbl = {1044.81723},
	zbmath = {1711981}}

@article{TakasakiTakebe95,
	eprint = {hep-th/9405096},
	archiveprefix = {arXiv},
	author = {Takasaki, Kanehisa and Takebe, Takashi},
	doi = {10.1142/S0129055X9500030X},
	fjournal = {Reviews in Mathematical Physics},
	issn = {0129-055X},
	journal = {Rev. Math. Phys.},
	language = {English},
	number = {5},
	pages = {743--808},
	title = {Integrable hierarchies and dispersionless limit},
	volume = {7},
	year = {1995},
	zbl = {0838.35117},
	zbmath = {796994}}

@article {Shiota,
    AUTHOR = {Shiota, Takahiro},
     TITLE = {Characterization of {J}acobian varieties in terms of soliton
              equations},
   JOURNAL = {Invent. Math.},
  FJOURNAL = {Inventiones Mathematicae},
    VOLUME = {83},
      YEAR = {1986},
    NUMBER = {2},
     PAGES = {333--382},
      ISSN = {0020-9910,1432-1297},
   MRCLASS = {14H40 (58F07)},
  MRNUMBER = {818357},
MRREVIEWER = {Bert\ van Geemen},
       DOI = {10.1007/BF01388967}
}

@misc{Zhou,
	title={Emergent Geometry and Mirror Symmetry of A Point}, 
	author={Jian Zhou},
	year={2015},
	eprint={1507.01679},
	archivePrefix={arXiv},
	primaryClass={math-ph}
}

@article {Konts,
    AUTHOR = {Kontsevich, Maxim},
     TITLE = {Intersection theory on the moduli space of curves and the
              matrix {A}iry function},
   JOURNAL = {Comm. Math. Phys.},
  FJOURNAL = {Communications in Mathematical Physics},
    VOLUME = {147},
      YEAR = {1992},
    NUMBER = {1},
     PAGES = {1--23},
      ISSN = {0010-3616,1432-0916},
   MRCLASS = {32G15 (14H15 58F07 81T40)},
  MRNUMBER = {1171758},
MRREVIEWER = {Claude\ Itzykson},
       URL = {http://projecteuclid.org/euclid.cmp/1104250524},
       DOI = {10.1007/BF02099526},
}

@article {KMMMZ,
    AUTHOR = {Kharchev, S. and Marshakov, A. and Mironov, A. and Morozov, A.
              and Zabrodin, A.},
     TITLE = {Towards unified theory of {$2$}d gravity},
   JOURNAL = {Nuclear Phys. B},
  FJOURNAL = {Nuclear Physics. B. Theoretical, Phenomenological, and
              Experimental High Energy Physics. Quantum Field Theory and
              Statistical Systems},
    VOLUME = {380},
      YEAR = {1992},
    NUMBER = {1-2},
     PAGES = {181--240},
      ISSN = {0550-3213,1873-1562},
   MRCLASS = {81T40 (81T27)},
  MRNUMBER = {1186584},
MRREVIEWER = {Stephen-wei\ Chung},
       DOI = {10.1016/0550-3213(92)90521-C},
       URL = {https://doi.org/10.1016/0550-3213(92)90521-C},
             	eprint={hep-th/9201013},
       archivePrefix={arXiv},
}

@article {AdlervM,
    AUTHOR = {Adler, M. and van Moerbeke, P.},
     TITLE = {A matrix integral solution to two-dimensional {$W_p$}-gravity},
   JOURNAL = {Comm. Math. Phys.},
  FJOURNAL = {Communications in Mathematical Physics},
    VOLUME = {147},
      YEAR = {1992},
    NUMBER = {1},
     PAGES = {25--56},
      ISSN = {0010-3616,1432-0916},
   MRCLASS = {58F07 (17B68 81T40)},
  MRNUMBER = {1171759},
MRREVIEWER = {Claude\ Itzykson},
       URL = {http://projecteuclid.org/euclid.cmp/1104250525},
       DOI = {10.1007/BF02099527},
}

@article {AAH3,
    AUTHOR = {Alexandrov, Alexander},
     TITLE = {K{P} integrability of triple {H}odge integrals. {II}.
              {G}eneralized {K}ontsevich matrix model},
   JOURNAL = {Anal. Math. Phys.},
  FJOURNAL = {Analysis and Mathematical Physics},
    VOLUME = {11},
      YEAR = {2021},
    NUMBER = {1},
     PAGES = {Paper No. 24, 82},
      ISSN = {1664-2368,1664-235X},
   MRCLASS = {37K10 (14N10 14N35 33C80 81R10 81T32)},
  MRNUMBER = {4195117},
MRREVIEWER = {Hsian-Hua\ Tseng},
       DOI = {10.1007/s13324-020-00451-7},
       URL = {https://doi.org/10.1007/s13324-020-00451-7},
      	eprint={2009.10961},
      archivePrefix={arXiv},
      primaryClass={math-ph}
}

\end{document}